\documentclass[journal,twoside,web]{IEEEtran}
\usepackage{cite}
\usepackage{amsmath,amssymb,amsfonts}

\usepackage{graphicx}
\usepackage{hyperref}
\hypersetup{hidelinks=true}
\usepackage{textcomp}
\usepackage{svg}
\usepackage{multirow}  
\usepackage{booktabs}  
\usepackage{array}     
\usepackage{makecell}  
\usepackage{xcolor}
\usepackage{graphicx}       
\usepackage{bm}
\usepackage{amssymb,amsmath,amsfonts}
\usepackage{amsthm}
\usepackage{subcaption}
\usepackage{color}
\usepackage{algorithm}
\usepackage{algpseudocode}

\usepackage{cases}
\usepackage{multirow}

\usepackage{cancel}
\usepackage{changes}

\usepackage{todonotes}
\usepackage{mathrsfs}

\usetikzlibrary{positioning}
\usetikzlibrary{shapes,arrows}
\newtheorem{remark}{\bf Remark}
\newtheorem{assumption}{\bf Assumption}
\newtheorem{lemma}{\bf Lemma}
\newtheorem{definition}{\bf Definition}
\newtheorem{proposition}{\bf Proposition}
\newtheorem{problem}{\bf Problem}

\newtheorem{thm}{\bf Theorem}

\newcommand{\wxy}{\color{black}}
\newcommand{\wxyv}{\color{black}}
\newcommand{\hw}{\color{black}} 

\newcommand{\col}{\textnormal{col}}

\usepackage{caption}
\usepackage{todonotes}
\usepackage{siunitx}
\allowdisplaybreaks

\def\BibTeX{{\rm B\kern-.05em{\sc i\kern-.025em b}\kern-.08em
    T\kern-.1667em\lower.7ex\hbox{E}\kern-.125emX}}
\begin{document}
\title{Reinforcement learning in pursuit-evasion differential game: safety, stability and robustness}

\author{Xinyang Wang, Hongwei Zhang, Jun Xu, Shimin Wang, Martin Guay 
\thanks{
This work has been submitted to the IEEE for possible publication. Copyright may be transferred without notice, after which this version may no longer be accessible.
This work was supported in part by the National Key Research and Development Program of China under Grant 2025YFE0101100 and in part by the National Natural Science Foundation of China under Grant 62473114. (Corresponding author: Hongwei Zhang)}
\thanks{Xinyang Wang, Hongwei Zhang and Jun Xu are with the School of Intelligence Science and Engineering, Harbin Institute of Technology, Shenzhen, Guangdong 518055, China (e-mail: wangxy@stu.hit.edu.cn; hwzhang@hit.edu.cn; xujunqgy@hit.edu.cn).}
\thanks{Shimin Wang is with Massachusetts
Institute of Technology, Cambridge, MA 02139, USA (e-mail: bellewsm@mit.edu).}
\thanks{Martin Guay is with Queen's University, Kingston, ON K7L 3N6, Canada (e-mail: guaym@queensu.ca).}
}

\maketitle

\begin{abstract}
Safety and stability are two critical concerns in pursuit-evasion (PE) problems in an obstacle-rich environment.
%
{\hw Most existing works combine control barrier functions (CBFs) and reinforcement learning (RL) to provide an efficient and safe solution. However, they do not consider the presence of disturbances, such as wind gust and actuator fault, which may exist in many practical applications.}
%
%
This paper integrates CBFs and a sliding mode control (SMC) term into RL to simultaneously address safety, stability, and robustness to disturbances. However, this integration is significantly challenging due to the strong coupling between the CBF and SMC terms.
Inspired by Stackelberg game, we {\wxy handle the coupling issue by proposing} a hierarchical design scheme where SMC and safe control terms interact with each other in a leader-follower manner. Specifically, the CBF controller, acting as the leader, enforces safety independently of the SMC design; while the SMC term, as the follower, is designed based on the CBF controller.
We then formulate the PE problem as a zero-sum game and propose a safe robust RL framework to learn the min-max strategy online.
{\wxy A sufficient condition is provided} under which the proposed algorithm remains effective even when  constraints are conflicting.
Simulation results demonstrate the effectiveness of the proposed safe robust RL framework.

\end{abstract}


\section{Introduction}
Pursuit-evasion (PE) tasks have been extensively investigated in various practical applications, such as air combat, intelligent transportation navigation, and surveillance operations \cite{liang2022a}, \cite{xiao2024learning}.
In a typical PE scenario, the evader seeks to escape from the pursuer, while the pursuer aims to prevent the evader from doing so.
Recently, differential game theory has been employed in PE scenarios to formulate PE interactions as two-player zero-sum games, enabling the pursuer to optimize its strategy against the worst-case actions of the evader \cite{weintraub2020an}.
Such formulations allow for the derivation of min-max strategies that provide guaranteed performance for the pursuer, irrespective of the evader’s actual behavior.

In PE differential games, the evader may deliberately induce the pursuer to enter dangerous regions, {\wxy such as obstacles}, {\hw to prevent being captured}.
Therefore, safety must be taken into account in the decision-making process of the pursuer.
%
%
{\hw Recently,} safety-aware PE games have been investigated in \cite{qu2023pursuit}, \cite{sun2023cooperative} and \cite{liu2025observer}, where safety constraints are directly incorporated into the optimization {\wxy problem} formulation.
However, this approach often imposes significant computational demands, posing challenges for real-time onboard implementation, especially for agents with limited computational resources.
%
{\hw A promising approach to provide computationally efficient solutions for control problems with safety requirements is control barrier functions (CBFs) \cite{ames2014control}, which are commonly used to formulate quadratic programs (QPs), also referred to as safety filters \cite{wabersich2023data}, that {\wxy slightly} modify potentially unsafe nominal control policies to enforce safety.}
This technique offers a feasible way to solve safety-aware PE games, i.e., design an unconstrained nominal (min-max) strategy to stabilize the PE system against unknown evader inputs, along with a safe control policy, projecting unsafe policies into the safe region.
%

%
Reinforcement learning (RL) offers great convenience in deriving nominal solutions with guaranteed performance for differential games \cite{kiumarsi2017optimal}. 
Recent studies \cite{zhang2024fixed}, \cite{dong2024adaptive}, \cite{wang2024nash} and \cite{tan2025unmanned} have successfully applied RL techniques to PE scenarios.
Yet, these studies fail to take safety issues into account. 
{\hw To address this limitation, by incorporating a CBF term into the reward function, a safety-oriented extension of RL (i.e., safe RL), 
%
was developed to seek the Nash equilibrium of PE games while avoiding collisions with multiple obstacles \cite{kololakis2024safety}.}
{\hw Although safety can be eventually guaranteed by the learned policy, it may still be violated during the learning process.}
To overcome this limitation, a new class of safe RL with strict safety guarantees was proposed in \cite{max2023safe}, where a safeguard policy was developed to decouple the safety constraints from the learning process, {\hw enforcing 
safety even if {\hw the learning process} fails to converge. Very recently, this approach has been successfully applied in human-robot collaboration tasks \cite{li2024optimal}, human-swarm interaction scenarios \cite{li2024game}, and safety-aware PE games \cite{jia2025safety}.
}
{\wxy 
Yet, this approach can only address safety constraints of relative degree one, and thus is not well suited for many PE games, 
where safety constraints often exhibit high relative degree w.r.t the pursuer’s input.
To handle such case, our previous work \cite{wang2025learning} proposes a high-order safeguard policy to enforce safety for high-relative-degree systems.
However, it focuses only on time-invariant safety specifications.
In practice, many systems, such as unmanned vehicles and manipulators, require time-varying safety guarantees, e.g., collision avoidance with moving obstacles, which further complicates the problem.
}
%

While safe RL shows great potential for addressing safety-aware PE differential games, most existing works \cite{zhang2024fixed}, \cite{dong2024adaptive}, \cite{wang2024nash}, \cite{tan2025unmanned}, \cite{kololakis2024safety} and \cite{jia2025safety}  
did not consider disturbances in system dynamics.
However, disturbances, such as modeling errors and external perturbations, are inevitable in practical scenarios, and can adversely affect both stability and safety of the system.
For example,   disturbances may cause the pursuer to deviate from its planned trajectory, or violate safety constraints such as collision avoidance. 
{\hw Very recently,} a safe robust RL was proposed in \cite{xiong2025time} to solve PE tasks with safety guarantees in the presence of disturbances.
However, it suffers from the same {\hw safety} issue as that in \cite{kololakis2024safety}.
{\hw To the best of our knowledge, enhancing the robustness in safe RL from both stability and safety point of view is a very interesting yet challenging research topic that has been insufficiently investigated so far}.
%
%


%
%

In this paper, we aim to propose a safe robust RL framework to learn a safe and robust min-max strategy for PE differential games subject to disturbances and unknown inputs of the evader.
To enhance the robustness of stability, we integrate sliding mode control (SMC) technique into safe RL, which has not been explored in existing SMC-based RL works \cite{yang2020event}, \cite{wang2024homotopic}.
To improve the robustness of safety, we develop a robust safe control policy within safe RL based on robust CBFs \cite{xiao2024robust}, a class of CBF techniques not yet applied in the existing safe RL literature \cite{kololakis2024safety}, \cite{max2023safe}, \cite{li2024optimal}, \cite{li2024game}.
%
%
%
%
However, {\hw integration of SMC and robust CBFs within safe RL} poses significant technical challenges due to their strong coupling in the controller design (see Sec. \ref{sec:safe_control}).
%
%
Inspired by Stackelberg game \cite{xu2022distributed}, we address the coupling issue between SMC and CBFs by designing a hierarchical control structure, where the safe control serves as the leader to enhance safety independent of the SMC; while the SMC, acting as a follower, {\wxy keeps} the system on a safe sliding surface.

The contributions of this paper are as follows:
\begin{enumerate}
    \item 
    Compared with related works \cite{zhang2024fixed}, \cite{dong2024adaptive}, \cite{wang2024nash}, \cite{tan2025unmanned}, \cite{kololakis2024safety} and \cite{jia2025safety}, which do not simultaneously consider  safety and disturbances, this paper studies safety-aware PE games under disturbances.
    The min-max strategy and safe control policy are designed in a robust manner using adaptive SMC and robust CBFs, without requiring prior knowledge of disturbance bounds.

    \item 
    {\wxy
    This paper extends our recent work \cite{wang2025learning} to handle time-varying safety constraints of high relative degree with partially unknown dynamics, which cannot be handled by existing safe RL methods \cite{max2023safe}, \cite{li2024optimal}, \cite{li2024game} and \cite{wang2025learning}.}
    %
    And, in contrast to \cite{kololakis2024safety}, \cite{jia2025safety} and \cite{xiong2025time} which assume static obstacles, our safe robust RL algorithm enables the pursuer to avoid collisions with moving obstacles, thus addressing a more practical yet challenging scenario.
    
    \item 
    Unlike existing safeguard works \cite{max2023safe}, \cite{li2024optimal}, \cite{li2024game} and \cite{jia2025safety} which only consider constraints of relative degree one, our proposed safeguard policy can  handle multiple constraints with mixed relative degrees. 
    Furthermore, we provide a sufficient condition to guarantee the effectiveness of the safeguard policy in handling potentially conflicting constraints.
    Interestingly, this condition aligns with those required for the feasibility of QP formulations using CBFs \cite{ames2017control}, but with less conservativeness.
\end{enumerate}

\emph{\textbf{Notations:}} 
Matrix $P\succ0$
means $P$ is positive definite.
The symbol $0$ denotes a zero vector with appropriate dimension.
Both Euclidean norm of a vector and Frobenius norm of a matrix are denoted by $\lVert \cdot \rVert$.
The maximum and minimum eigenvalues of a matrix are denoted by $\bar{\sigma}(\cdot)$ and $\underline{\sigma}(\cdot)$, respectively.
For a time-varying bounded signal $d(t)$, $\lVert d \rVert_\infty = \sup_{t \geq 0} \lVert d(t) \rVert$.
%

\section{Preliminaries and problem formulation} \label{Sec Formulation}
\subsection{High-order control barrier function}
Consider a nonlinear system 
\begin{align} \label{eq:nominal_sys}
    \dot{x} = f(x)+g(x)u,
\end{align}
where $x \in \mathbb{R}^n$ and $u\in \mathbb{R}^p$ are the state  and the control input of the system, respectively; $f: \mathbb{R}^n \rightarrow \mathbb{R}^n$ and   $g: \mathbb{R}^n \rightarrow \mathbb{R}^{n\times p}$ are locally Lipschitz.

Consider a safe set described as $$\mathscr{C} = \{x\in \mathbb{R}^n :b(x)\geq 0\},$$ where $b$ is $m^{th}$-order continuously differentiable.
The relative degree of $b$ for \eqref{eq:nominal_sys} is defined as:
\begin{definition}\label{def:IRD} 
(relative degree \cite{khalil2002nonlinear}):
Consider a safe set $\mathscr{C}$.
The relative degree of $b(x)$ on $\mathscr{C}$ w.r.t system \eqref{eq:nominal_sys} is $m$ if $L_gL_f^{k} b(x) =0, ~\forall k\in \{0,1,\cdots,m-2\}$ and $L_gL_f^{m-1} b(x) \neq 0 \text{ for all } x \in \mathscr{C}$, where $L_f b$ and $L_g b$ are Lie derivatives of $b$ along $f$ and $g$, respectively.
\end{definition}

We define a series of sets for $i\in \{1,2,\cdots,m\}$:
\begin{align} \label{eq:classical_HOCBF_set}
    \mathscr{C}_i = \{ x \in \mathbb{R}^n:~\theta_{i-1}(x)\geq 0\}, \;\; \mathscr{C} = \cap \mathscr{C}_i   
\end{align}
with $\theta_i$ being defined as
\begin{align} \label{eq:hocbf}
    \theta_i(x) = \dot{\theta}_{i-1}(x)+\alpha_i(\theta_{i-1}(x)), \;\; \theta_0(x) = b(x),
\end{align}
where $\alpha_{i}$ is an $(m-i)^{th}$ order continuously differentiable class $\mathcal{K}$ function.
Then the high-order control barrier function (CBF) can be defined as:
\begin{definition}\label{def:HOCBF} 
(high-order CBF \cite{xiao2021high}):
Consider a safe set $\mathscr{C}$.
An $m^{th}$-order continuously differentiable function $b$ is a high-order CBF of relative degree $m$ for system \eqref{eq:nominal_sys} if there exists a class $\mathcal{K}$ function $\alpha$ such that
\begin{align*}
    \mathop{\rm{sup}}\limits_{u \in \mathbb{R}^p} \{ L^{m}_{f} b & + L_gL^{m-1}_{f}b~u \} \geq -  \mathcal{O}(x,t) - \frac{\partial^m b}{\partial t^m} -\alpha(\theta_{m-1})
\end{align*}
for all $(x,t) \in \mathscr{C} \times[0,\infty)$, where $$\mathcal{O} = \sum_{i=1}^{m-1} L_f^i(\alpha_{m-1} \circ \theta_{m-i-1}) + \frac{\partial^i(\alpha_{m-1} \circ \theta_{m-i-1})}{\partial t^i}$$ and $\circ$ denotes the composition of functions.
\end{definition}

\subsection{Pursuit-evasion differential game}
Consider the safety-aware pursuit-evasion (PE) differential game consisting of a pursuer and an evader, where the dynamics of the pursuer are
\begin{align} \label{eq:pursuer_dynamics}
    \dot{z}_p  = f_p\left(z_p\right) + g_p(z_p)(u_p + d),
\end{align}
where $z_p = \col(x_p,v_p)$; $x_p \in \mathbb{R}^n$, $v_p \in \mathbb{R}^n$, $u_p \in \mathbb{R}^p$ and $d \in \mathbb{R}^p$ denote the position, velocity, control input, and unknown disturbance of the pursuer, respectively.
The functions $f_p : \mathbb{R}^{2n} \rightarrow \mathbb{R}^{2n}$ and $g_p : \mathbb{R}^{2n} \rightarrow \mathbb{R}^{2n\times p}$ are locally Lipschitz.
It is assumed that $f_p(0) = 0$, $g_p(z_p)$ is full column rank, and the system is stabilizable.

{\wxyv 
The dynamics of the evader are
\begin{align*}
    \dot{z}_e = f_e(z_e) + g_e (z_e) u_e,
\end{align*}
where $z_e = \col(x_e,v_e)$; $x_e \in \mathbb{R}^n$, $v_e \in \mathbb{R}^n$ and $u_e \in \mathbb{R}^p$ are the position, velocity and control input of the evader, respectively.
The functions $f_e: \mathbb{R}^{2n} \rightarrow\mathbb{R}^{2n}$ and $g_e:\mathbb{R}^{2n} \rightarrow\mathbb{R}^{2n \times p}$ are locally Lipschitz.
%
%

Define the relative state with respect to the pursuer and the evader as $\xi = z_p-z_e$. Then its dynamics are 
\begin{align} \label{eq:pursuer_evader_dynamics}
    \dot{\xi} = \mathcal{F}(z_p,z_e) + g_p(z_p)(u_p+d)  - g_e(e_e)u_e,
\end{align}
where $\mathcal{F}(z_p,z_e)=f_p(z_p) - f_e(z_e)$.
}

The pursuer is required to avoid collisions with moving obstacles and satisfy velocity constraints.
The safe set for position and velocity are defined as
\begin{align} \label{eq:safetyset}
    &\mathscr{C}_{x} = \{ z_p \in \mathbb{R}^{2n}:~ b_x(z_p,x_o) \geq 0\}, \notag\\
    &\mathscr{C}_{v} = \{ z_p \in \mathbb{R}^{2n}:~ b_v(z_p) \geq 0\},
\end{align}
where $b_x:\mathbb{R}^{2n}\times\mathbb{R}^n \rightarrow \mathbb{R}$ and $b_v:\mathbb{R}^{2n} \rightarrow\mathbb{R}$ are sufficiently differentiable functions, and $x_o$ is the position of the obstacle.
It is assumed that $b_x$ has relative degree two and $b_v$ has relative degree one w.r.t \eqref{eq:pursuer_dynamics}, which implies $\lVert L_{g_p}L_{f_p}b_x \rVert \neq 0$ and $\lVert L_{g_p}b_v \rVert \neq 0$.
Then $\dot{b}_x$ and $\ddot{b}_x$ can be expressed as
\begin{align*}
    &\dot{b}_x = L_{f_p}b_x+\frac{\partial b_x}{\partial x_o}^\top  \frac{\partial x_o}{\partial t}, \\
    &\ddot{b}_x =L^2_{f_p}b_x+L_{g_p}L_{f_p}b_x (u+d)+ \frac{\partial L_{f_p}b_x}{\partial x_o}^\top\frac{\partial x_o}{\partial t} \\
    &~~~~~~~~+\left(\frac{\partial^2 b_x}{\partial x_o^2}^\top \frac{\partial x_o}{\partial t} + L_{f_p} \frac{\partial b_x}{\partial x_o}\right)^\top\frac{\partial x_o}{\partial t} + \frac{\partial b_x}{\partial x_o}^\top  \frac{\partial^2 x_o}{\partial t^2}.
\end{align*}
If $d$, $\partial x_o/\partial t$ and $\partial^2 x_o/\partial t^2$ are known, the collision avoidance can be enforced using a high-order CBF.
{\wxy
While $d$ and $\partial^2 x_o/\partial t^2$ are unknown, the collision avoidance can also be enforced using our previous work \cite{wang2025learning} if $\partial x_o/\partial t$ is known for control design.
However, this information is generally difficult to obtain in practice, which significantly complicates the safety-aware PE games.} 
Therefore, we study the case where $d$, $\partial x_o/\partial t$ and $\partial^2 x_o/\partial t^2$ are unknown, and impose the following assumptions:

\begin{assumption} \label{as:disturbance}
    The disturbance $d$ is bounded by an unknown positive real number $\bar{d}$, i.e., $\lVert d\rVert_\infty \leq \bar{d}$. 
\end{assumption}

\begin{assumption}\label{as:obstacle}
    The obstacle's position $x_o$ is bounded and its velocity satisfies $\lVert \frac{\partial x_o}{\partial t}\rVert \leq \eta$, where $\eta>0$ is known. 
\end{assumption}

Before presenting the PE differential game, we first make a general assumption regarding the evader:
\begin{assumption} \label{as:u_e}
    The evader's state $z_e \in \Omega_e$ for a compact set $\Omega_e$, and
    its control input $u_e \in L_2 \cap L_\infty$, i.e., $$\int_0^\infty \lVert u_e(\tau) \rVert^2 d\tau < \infty,~~~ \lVert u_e\rVert_\infty \leq \bar{u}_e,$$ where $\bar{u}_e$ is an unknown positive real number. 
\end{assumption}

We now formulate the safety-aware PE problem studied in this paper:

\begin{problem} \label{pro}
    Consider the systems \eqref{eq:pursuer_dynamics} and \eqref{eq:pursuer_evader_dynamics}, together with the safe sets $\mathscr{C}_x$ and $\mathscr{C}_v$ defined in \eqref{eq:safetyset}.
    Under Assumptions \ref{as:disturbance}-\ref{as:u_e}, design a robust safe control policy $u_p$ such that the PE task is accomplished, i.e., $\lim_{t\rightarrow \infty} \lVert  \xi(t) \rVert = 0,$ despite the presence of unknown disturbance $d$ and the evader's input $u_e$ while guaranteeing the forward invariance of the safe sets, i.e., $z_p(t) \in \mathscr{C}_{x}\cap \mathscr{C}_{v} , ~\forall t\geq 0$.
\end{problem}

To address Problem \ref{pro}, we propose a safe robust RL framework that simultaneously accounts for stability, robustness and safety.

\section{Solutions for safety-aware PE game}
In this section, we design a three-term controller for the pursuer as $$u_p = u_n + u_r + u_s,$$ where $u_n$ denotes the unconstrained nominal (min-max) strategy for accomplishing the PE task under unknown $u_e$; $u_r$ is the robust control policy for rejecting disturbance $d$; and $u_s$ represents the safeguard policy for guaranteeing safety.

\subsection{Design of the min-max strategy $\mathnormal{u_n}$}
To obtain a min-max strategy $u_n$, we formulate a two-player zero-sum PE game of the undisturbed nominal system
\begin{align} \label{eq:pursuer_evader_dynamics_nodis}
    \dot{\xi} = \mathcal{F}(z_p,z_e) + g_p(z_p)u_n - g_e(z_e)u_e,
\end{align}
with the following performance index
\begin{align} \label{eq:performance_index}
    J(\xi,u_n,u_e) = \int_0^\infty l(Q(\xi(\tau)), u_n(\tau),u_e(\tau)) d\tau
\end{align}
and the reward function
\begin{align*}
    l(\xi,u_n,u_e) =Q(\xi) + u_n^\top R u_n - \gamma^2 u_e^\top T u_e
\end{align*}
where $Q:\mathbb{R}^n\rightarrow\mathbb{R}$ is a positive definite function, $R \succ 0,~ T \succ 0$ and $\gamma>0$ is a prescribed gain.   

The main goal is to find an admissible feedback control $u_n$ that ensures the $L_2$ stability of \eqref{eq:pursuer_evader_dynamics_nodis}, i.e.,
\begin{align} \label{eq:disturbance_attenuation}
    \int_0^\infty (Q(\xi) + u_n^\top Ru_n)d\tau \leq \gamma^2\int_0^\infty  u_e^\top T u_e d\tau,
\end{align}
where $\gamma \geq \gamma^\star$ and  {\hw $\gamma^\star$ is the minimum value of $\gamma$ such that disturbance attenuation condition \eqref{eq:disturbance_attenuation} can be achieved.}

Now, we can solve the PE problem by seeking the Nash equilibrium of the following min-max optimization problem
\begin{align*}
    \min_{u_n \in \mathbb{R}^p} \max_{u_e \in \mathbb{R}^p}J(\xi,u_n,u_e).
\end{align*}
The Nash condition is defined as
\begin{definition} \cite{lewis2012optimal}
$\{u_n^\star, u_e^\star\}$ is the Nash solution of the zero-sum game if 
the following condition holds:
\begin{align*}
J(\xi,u_n^\star,u_e) \leq  J(\xi,u_n^\star,u_e^\star) \leq J(\xi,u_n,u_e^\star).
\end{align*}
\end{definition}

For such an infinite-horizon optimal control problem, we define the value function as
\begin{align*}
    V(\xi(t)) = \int_{t}^\infty ( Q(\xi) +  u_n^\top R u_n - \gamma^2u_e^\top T u_e) d\tau,
\end{align*}
and its corresponding Hamilton function
\begin{align*}
    H(\xi, \nabla V, u_n,u_e) =  l(\xi,u_n,u_e)+ \nabla V^\top(\mathcal{F} + g_p u_n -g_eu_e),
\end{align*}
where $\nabla V = \partial V/\partial \xi$.
Define the optimal value of $V$ as
\begin{align*}
    V^\star(\xi) = \min_{u_n \in \mathbb{R}^p} \max_{u_e \in \mathbb{R}^p}\int_{t}^\infty ( Q(\xi) +  u_n^\top R u_n - \gamma^2u_e^\top T u_e) d\tau,
\end{align*}
and it is the unique positive definite solution to the Hamilton-Jacobi-Isaacs (HJI) equation
\begin{align} \label{eq:HJB}
    \min_{u_n \in \mathbb{R}^p} \max_{u_e \in \mathbb{R}^p} H(\xi,\nabla V^\star,u_n,u_e)  = 0,
\end{align}
where $\nabla V^\star = \partial V^\star/\partial\xi$ and $V^\star(0) = 0$.
Let $u_n^\star$ and $u_e^\star$ denote the optimal control polices.
According to the stationary condition \cite{lewis2012optimal}, we obtain
\begin{align*}
    &\frac{\partial H(\xi,\nabla V^\star,u_n,u_e)}{\partial u_n}\bigg|_{u_n = u_n^\star} = 0, \\
    &\frac{\partial H(\xi,\nabla V^\star,u_n,u_e)}{\partial u_e}\bigg|_{u_e = u_e^\star} = 0.
\end{align*}
Then the optimal control policies can be derived as
\begin{align}  \label{eq:optimal_control_policy}
    u_n^\star = -\frac{1}{2}R^{-1} g_p^\top \nabla V^\star, ~~u_e^\star =  -\frac{1}{2\gamma^2}T^{-1} g_e^\top \nabla V^\star. 
\end{align}
Substituting \eqref{eq:optimal_control_policy} into \eqref{eq:HJB}, the HJI equation becomes
\begin{align} \label{eq:HJB_equation}
    0 = &\nabla V^{\star\top}\mathcal{F} - \frac{1}{4} \nabla V^{\star\top}g_pR^{-1}g_p^\top \nabla V^\star \notag\\
    &~~~~~~~~~~~~~~+Q(\xi)+\frac{1}{4\gamma^2} \nabla V^{\star\top}g_eT^{-1}g_e^\top \nabla V^\star.
\end{align}

\begin{lemma}
    Suppose Assumption \ref{as:u_e} holds.
    Let $V^\star$ be a continuously differentiable positive definite solution of \eqref{eq:HJB} with the control policies $u_n$ and $u_e$ taking the optimal forms \eqref{eq:optimal_control_policy}.
    The PE system \eqref{eq:pursuer_evader_dynamics_nodis} is $L_2$ stable with $L_2$-gain bounded by $\gamma$.
\end{lemma}

The proof is similar to that of Theorem 1 in \cite{li2025robust}, and is omitted due to space limitation.

\subsection{Design of the adaptive robust control $\mathnormal{u_r}$}
Design the following integral safe sliding surface
{\wxyv
\begin{align}\label{eq:sliding_manifold}
    s(t) =&  S(z_p(t)) - S(z_p(0)) \notag \\
    &- \int_0^t  \chi(z_p)\big(f_p(z_p)+g_p(z_p)(u_s+u_n)\big)d\tau,
\end{align}
where $S\in\mathbb{R}^m$, $\chi= \partial S/\partial z_p$, and $u_s$ is the safe control policy to be designed in Section \ref{sec:safe_control}.}
The dynamics of the integral safe sliding surface is 
\begin{align*}
    \dot{s} = \chi(z_p )\dot{z}_p(t) - \chi(z_p)(f_p(z_p)+ g_p(z_p)(u_s+u_n)).
\end{align*}

%
%
Then we design the robust control policy $u_{r}$ to keep the state on the integral safe sliding manifold $s = 0$ as
\begin{align} \label{eq:robust_control_sign}
    u_{r} = - \mathcal{K}(t) \textnormal{sign} \left(g_p^\top(z_p) \chi^\top(z_p) s\right)
\end{align}
where $\mathcal{K}(t)$ is an adaptive gain.

To address the chattering issue, we replace $sign$ function with a continuous function $tanh$. Then \eqref{eq:robust_control_sign} becomes
\begin{align} \label{eq:robust_control_tanh}
    u_{r} = - \mathcal{K}(t) \tanh \left(\frac{g_p^\top(z_p) \chi^\top(z_p) s}{\rho}\right),
\end{align}
where $\rho$ is a very small positive constant, and the tuning law for $\mathcal{K}(t)$ is given as

(i) when $\lVert s(t) \rVert > \varepsilon>0$
\begin{align} \label{eq:sliding_adaptive_gain_update_s>0}
    \dot{\mathcal{K}}(t) = K_{1} + K_{2} \lVert s(t) \rVert
\end{align}
with $K_1,K_2 > 0$ and $\mathcal{K}(0) > 0$.

(ii) when $\lVert s(t) \rVert \leq \varepsilon$
\begin{align} \label{eq:sliding_adaptive_gain_update_s=0}
    \mathcal{K}(t) = K_3 + \mathcal{K}(t^*) \lVert \theta(t) \rVert \notag \\
    K_{4}\dot{\theta}(t) + \theta(t) = \textnormal{sign}(s(t))
\end{align}
with $K_3,K_4>0$, and $t^*$ is the largest time such that $\lVert s(t^*)\rVert \leq \varepsilon$ and $\lVert s({t^*}^-)\rVert > \varepsilon$; ${t^*}^-$ denotes the time instant just before $t^*$.

We now show that the PE error $z_p$ can reach the safe sliding surface in finite time.
The following two lemmas are introduced to facilitate the proof.

\begin{lemma} \cite{liu2025practical} \label{lm:tanh_sign}
    Given any $\rho>0$, the following inequality 
    \begin{align*}
        0 \leq x \,\textnormal{sign}(x)- x \tanh\left(\frac{x}{\rho}\right) \leq \kappa \rho
    \end{align*}
    always holds for all $x \in \mathbb{R}$, where $\kappa = 0.2785$.
\end{lemma}

\begin{lemma}\label{lm:K_upper_bound}
    Consider the system \eqref{eq:pursuer_dynamics} with safe sliding surface \eqref{eq:sliding_manifold} controlled by \eqref{eq:robust_control_tanh}, \eqref{eq:sliding_adaptive_gain_update_s>0} and \eqref{eq:sliding_adaptive_gain_update_s=0}.
    Suppose Assumption \ref{as:disturbance} holds.
    Then there exists a positive real number $\mathcal{K}^*$ such that $0 < \mathcal{K}(t) \leq \mathcal{K}^*$.
\end{lemma}

The proof of Lemma \ref{lm:K_upper_bound} is similar to that of Lemma 1 in \cite{plestan2010new}, and is omitted here due to space limitation.

The following theorem shows that the safe sliding surface can be reached in finite time.
\begin{thm} \label{thm:sliding_control}
    Consider the disturbed system \eqref{eq:pursuer_dynamics} with safe sliding surface defined by \eqref{eq:sliding_manifold}.
    Suppose Assumption \ref{as:disturbance} holds.
    Let the robust control $u_r$ be designed as in \eqref{eq:robust_control_tanh}, \eqref{eq:sliding_adaptive_gain_update_s>0}, and \eqref{eq:sliding_adaptive_gain_update_s=0} with $\chi(z_p) = g_p^\dagger(z_p)$.
    Then $u_r$ drives the system trajectories to reach the safe sliding surface $s =0$ in finite time and stay on it thereafter.
\end{thm}

\begin{proof}
    See Appendix \ref{Appendix_sliding_control}
\end{proof}

\subsection{Design of the safeguard control $\mathnormal{u_s}$} \label{sec:safe_control}
%
{\wxy In this section, we employ a robust CBF technique to derive a safeguard policy that enforces safety constraints without knowing the obstacle dynamics and disturbance.}
Specifically, under Assumption \ref{as:obstacle}, we have
\begin{align*}
    \dot{b}_x(z_p,x_o) =& L_{f_p} b_x + \frac{\partial b_x}{\partial x_o}^\top \frac{\partial x_o}{\partial t} \\
     \geq & L_{f_p} b_x - \frac{1}{4k}\left\lVert \frac{\partial b_x}{\partial x_o}\right\rVert^2 - k \eta^2 := \tilde{b}_x(z_p,x_o),
\end{align*}
where $k$ is an arbitrary positive real number.
Define the following functions 
\begin{align} \label{eq:HO-DRCBF}
    \phi_{x,0}(z_p,x_o) &= b_x(z_p,x_o), \notag\\
    \phi_{x,1}(z_p,x_o) &= \tilde{b}_x(z_p,x_o)+ p_1\phi_{x,0}(z_p,x_o),
\end{align}
and the corresponding sets
\begin{align*}
\mathscr{C}_{x,0} = \{ z_p \in \mathbb{R}^{2n} : \phi_{x,0} \geq 0\}, \mathscr{C}_{x,1} = \{ z_p \in \mathbb{R}^{2n} :\phi_{x,1} \geq 0\},
\end{align*}
where $p_1$ is a positive real number.

For the position constraint, the CBF condition is $\dot{\phi}_{x,1} + p_2\phi_{x,1} \geq 0$, where $p_2$ is a positive number (see Definition \ref{def:HOCBF}).
The following lemma provides a sufficient condition for $u_p$ to guarantee collision avoidance with a moving obstacle of unknown dynamics.
\begin{lemma} \label{lm:robust_cbf}
    Given functions $b_x$ and $\tilde{b}_x$, define $\phi_{x,0}$ and $\phi_{x,1}$ as in \eqref{eq:HO-DRCBF} with the associated sets $\mathscr{C}_{x,0}$ and $\mathscr{C}_{x,1}$.
    Let $p_1, p_2$ be two positive real numbers such that $z_p(0) \in \mathscr{C}_{x,0} \cap \mathscr{C}_{x,1}$.
    Let $K_{x}$ be defined as 
    \begin{align*}
        K_{x} = \big\{ u_p \in \mathbb{R}^p:~ \breve{\alpha}_1(z_p,t) + \breve{\beta}_1(z_p,t)(u_p+d) \geq 0 \big\},
    \end{align*}
    where $\breve{\alpha}_1(z_p,t) = L_{f_p} \phi_{x,1} + \frac{\partial \phi_{x,1}}{\partial x_o}^\top \frac{\partial x_o}{\partial t}+p_2 \phi_{x,1}$ and $\breve{\beta}_1(z_p,t) = L_{g_p} \phi_{x,1}$.
    Then any locally Lipschitz controller $u_p \in K_{x}$ renders the set $\mathscr{C}_{x,0} \cap \mathscr{C}_{x,1}$ forward invariant.
\end{lemma}

\begin{proof}
    If $u_p \in K_x$, then $\dot{\phi}_{x,1} + p_2 \phi_{x,1} \geq 0$ holds.
    According to Theorem 2 in \cite{xiao2021high}, one has $\phi_{x,1}(z_p(t),x_o(t)) \geq 0~~ \forall t \geq 0$ since $\phi_{x,1}(z_p(0),x_o(0)) \geq 0$.
    Since $\dot{\phi}_{x,0} + p_1 \phi_{x,0} = L_{f_p}b_x + \frac{\partial b_x}{\partial x_o}^\top \frac{\partial x_o}{\partial t} + p_1 \phi_{x,0} \geq \phi_{x,1} $, $\phi_{x,1} \geq 0$ implies that $ \dot{\phi}_{x,0} + p_1 \phi_{x,0} \geq 0$.
    Hence, $\phi_{x,0}(z_p(t),x_o(t)) \geq 0$ holds for all $t \geq 0$ since $\phi_{x,0}(z_p(0),x_o(0)) \geq 0$.
    Then one can conclude that $\mathscr{C}_{x,0} \cap \mathscr{C}_{x,1}$ is forward invariant under $u_p \in K_x$.
\end{proof}

For the velocity constraint, the CBF condition is $\dot{b}_v + \tilde{p}b_v \geq 0$, where $\tilde{p}>0$ is a positive real number (see Definition \ref{def:HOCBF}).
For brevity, we denote $\breve{\alpha}_2(z_p) = L_{f_p}b_v + \tilde{p}b_v$ and $\breve{\beta}_2(z_p) =L_{g_p}b_v$.
Then any locally Lipschitz control policy $u_p \in K_v$, where $K_v$ is defined as 
\begin{align*}
    K_{v} = \big\{ u_p \in \mathbb{R}^p:~ \breve{\alpha}_2(z_p) + \breve{\beta}_2(z_p)(u_p+d) \geq 0  \big\},
\end{align*}
can render $\mathscr{C}_v$ forward invariant.

To simultaneously enforce collision avoidance and velocity constraint under $u_p$, we design the safeguard policy by addressing the following Quadratic Program (QP) as
\begin{align*}
    u_s = &\textnormal{arg}\min_{\mu \in \mathbb{R}^p} \frac{1}{2}\lVert \mu \rVert^2 \\
    &\textnormal{s.t.} ~~\breve{\alpha}_1(z_p,t) + \breve{\beta}_1(z_p,t)(u_r+u_n+\mu+ d)\geq 0, \\
    & ~~~~~~\breve{\alpha}_2(z_p) + \breve{\beta}_2(z_p)(u_r+u_n+\mu+ d) \geq 0.
\end{align*}

Define the following Lagrangian function
\begin{align*}
    L(z_p,u_s,&\lambda_x,\lambda_v) =  \frac{1}{2} \lVert u_s\rVert^2  \\
    &- \lambda_x \left(\breve{\alpha}_1(z_p,t) + \breve{\beta}_1(z_p,t)(u_r+u_n+u_s+ d)\right) \\
    & - \lambda_v \left(\breve{\alpha}_2(z_p) + \breve{\beta}_2(z_p)(u_r+u_n+u_s+ d)\right),
\end{align*}
where $\lambda_x$ and $\lambda_v$ are Lagrange multipliers w.r.t position constraint and velocity constraint, respectively.

Applying the Karush-Kuhn-Tucker (KKT) conditions 
\begin{align} \label{eq:KKT}
        &\partial L/ \partial u_s = 0, ~~ \lambda_x \geq 0,~~ \lambda_v \geq 0,\notag \\
        &\lambda_x \left(\breve{\alpha}_1(z_p,t) + \breve{\beta}_1(z_p,t)(u_r+u_n+u_s+ d)\right) = 0, \notag\\
        &\lambda_v \left(\breve{\alpha}_2(z_p) + \breve{\beta}_2(z_p)(u_r+u_n+u_s+ d)\right) = 0, \notag\\
        &\breve{\alpha}_1(z_p,t) + \breve{\beta}_1(z_p,t)(u_r+u_n+u_s+ d) \geq 0,  \notag\\
        &\breve{\alpha}_2(z_p) + \breve{\beta}_2(z_p)(u_r+u_n+u_s+ d) \geq 0, 
\end{align}
the safe control policy can be derived as
\begin{align} \label{eq:KKT_safe_control}
    u_s = \lambda_x\breve{\beta}_1^\top(z_p,t) + \lambda_v \breve{\beta}_2^\top(z_p),
\end{align}
where the multipliers $\lambda_x~\textnormal{and}~\lambda_v$ depend on the values of $x_p, v_p,x_o, \partial x_o/\partial t, u_r,u_n$ and $d$.\footnote{Interested readers can refer to \cite{tan2024on} for detailed form of $\lambda_x$ and $\lambda_v$.}

It is worth noting that there exists a strong coupling between $u_r$ and $u_s$; specifically, the design of $u_r$ depends on $u_s$ (see \eqref{eq:sliding_manifold}), while the computation of $u_s$ requires information from $u_r$ (see \eqref{eq:KKT}).
This coupling renders the direct combination of SMC and CBF technically infeasible.
Inspired by the Stackelberg game framework \cite{xu2022distributed}, we design $u_r$ and $u_s$ in a leader-follower sense,
{\wxy where $u_s$ acts as the leader, which is designed independently of $u_r$ to enforce safety, and $u_r$ acts as the follower, designed subsequently to keep the system staying on the safe sliding surface defined by $u_s$.}
%
Under this framework, we replace \eqref{eq:KKT_safe_control} with the following form:
\begin{align} \label{eq:safe_control}
    u_s = \zeta_x(z_p,t) \breve{\beta}_1^\top(z_p,t) + \zeta_v(z_p) \breve{\beta}_2^\top(z_p)
\end{align}
where $\zeta_x(z_p,t)$ and $ \zeta_v(z_p)$ are two energy functions that do not depend on $u_r$.
This design effectively resolves the coupling issue between $u_r$ and $u_s$.

%

%
Now we provide the following theorem to show that $u_s$ guarantees the safety of the PE game.
Furthermore, we analyze the relationship between the feasibility of QP and the effectiveness of $u_s$ when both position and velocity constraints are simultaneously active.
To make convenience in the safety analysis, we define the following sets 
\begin{align*}
    &\mathscr{C}^I_\cap(t) = \operatorname{Int}(\mathscr{C}_{x,1}(t)) \cap \operatorname{Int}(\mathscr{C}_v) ,\\
    &\mathscr{C}^S_{\cap }(t) =\partial\mathscr{C}_{x,1}(t) \cap \partial\mathscr{C}_v,\\
    &\mathscr{C}^S_{\cup }(t) =\partial\mathscr{C}_{x,1}(t) \cup \partial\mathscr{C}_v,
\end{align*}
where
\begin{align*}
    &\operatorname{Int}(\mathscr{C}_{x,1}) = \{ z_p \in \mathbb{R}^{2n}:~ \phi_{x,1}(z_p,x_o) >0\},\\
    &\operatorname{Int}(\mathscr{C}_{v}) = \{ z_p \in \mathbb{R}^{2n}:~ b_v(z_p) > 0\}, \\
    &\partial\mathscr{C}_{x,1}= \{ z_p \in \mathbb{R}^{2n}:~ \phi_{x,1}(z_p,x_o) =0\},\\
    &\partial\mathscr{C}_{v} = \{ z_p \in \mathbb{R}^{2n}:~ b_v(z_p) = 0\}.
\end{align*}

\begin{thm} \label{thm:safe_analysis}
    Suppose Assumptions \ref{as:disturbance}-\ref{as:u_e} hold, and $z_p \in \Omega_p$ for a compact set $\Omega_p \in \mathbb{R}^{2n}$.
    Let $u_s$ be defined in \eqref{eq:safe_control}, $u_r$ be defined in \eqref{eq:robust_control_tanh}, $u_n$ be a locally Lipschitz stabilizing controller, and $\zeta_x(z_p,t), \zeta_v(z_p)$ satisfy the following conditions
    \begin{align*} 
        &\zeta_x(z_p,t) = -\tilde{K}_x  \nabla B_x(\phi_{x,1}(z_p,t)),\\
        &\zeta_v(z_p) = -\tilde{K}_v  \nabla B_v(b_v(z_p)), \notag
    \end{align*}
    where $\tilde{K}_x,\tilde{K}_v > 0$ are bounded safe control gain, $B_x, B_v$ are user-defined energy functions, $\nabla B_x = \partial B_x/\partial \phi_{x,1}$ and $\nabla B_v = \partial B_v/\partial b_v$.
    If $z_p(0) \in \mathscr{C}^I_{\cap}(0)$, $B_x$ and $B_u$ satisfy the following inequalities
        \begin{align}
        \frac{1}{\gamma_1(\phi_{x,1}(z_p,t))}  \leq& B_x(\phi_{x,1}(z_p,t)) \leq \frac{1}{\gamma_2(\phi_{x,1}(z_p,t))} \label{bx_pro}\\
        \frac{1}{\gamma_3(b_v(z_p))}  \leq& B_v(b_v(z_p)) \leq \frac{1}{\gamma_4(b_v(z_p))}\label{bu_pro}
        \end{align}
    for some class $\mathcal{K}$ functions $\gamma_1, \gamma_2, \gamma_3$ and $\gamma_4$, and if for any $z_p \in \mathscr{C}_\cap^S(t)  ~~ \forall t \geq 0$, there exists a small neighborhood $N_l(z_p, t)$ such that the Gram matrix is positive definite, i.e.,
    \begin{align} \label{schur-Matrix}
        \begin{bmatrix}
            \breve{\beta}_1\breve{\beta}_1^\top  & \breve{\beta}_1\breve{\beta}_2 ^\top\\
            \breve{\beta}_2\breve{\beta}_1^\top &  {\wxyv \breve{\beta}_2\breve{\beta}_2^\top}
        \end{bmatrix}\succ 0~~~~ \forall x \in N_l(z_p, t),
    \end{align}
    Then the safe control policy $u_s$ renders the safe set $\mathscr{C}^I_{\cap}$ forward invariant, i.e., $\phi_{x,1}(z_p, t), b_v(z_p) > 0 ~~\forall t \geq 0$.
\end{thm}

\begin{proof}
    See Appendix \ref{Appendix_safe_analysis}.
\end{proof}

From Theorem \ref{thm:safe_analysis}, $\phi_{x,1}(z_p, t) > 0$ for all $t \geq 0$, which implies that $u_p \in K_x$.
Then one has $b_x \geq 0 ~\forall t \geq 0$ according to Lemma \ref{lm:robust_cbf}.
Hence, both position constraint and velocity constraint can be simultaneously enforced if \eqref{schur-Matrix} is satisfied.
Since $b_x(z_p(0),x_o(0)) > 0$ and $b_v(z_p(0)) > 0$, one has $b_x(z_p(t),x_o(t)) > 0$ and $b_v(z_p(t)) > 0$ for all $t \geq 0$ according to Theorem 2 in \cite{xiao2021high}.

\begin{remark}
    In Theorem \ref{thm:safe_analysis}, we provide a sufficient condition for guaranteeing safety under potentially conflicting position and velocity constraints, and establish the relationship between the feasibility of the QP solution \eqref{eq:KKT_safe_control} and the effectiveness of the safeguard policy \eqref{eq:safe_control} when both constraints are active.
    It is important to note that positive definiteness of the Gram matrix \eqref{schur-Matrix} is a necessary condition for the feasibility of the QP when position constraint and velocity constraint are in conflict.
    For example, consider the case where both position and velocity constraints are active.
    According to KKT conditions \eqref{eq:KKT}, the multipliers $\lambda_x$ and $\lambda_v$ w.r.t \eqref{eq:KKT_safe_control} are calculated as
    \begin{align} \label{eq:lambda_both_active}
        \begin{bmatrix}
            \lambda_x \\ \lambda_v
        \end{bmatrix}
        =
        \begin{bmatrix}
            \breve{\beta}_1\breve{\beta}_1^\top  & \breve{\beta}_1\breve{\beta}_2 ^\top\\
            \breve{\beta}_2\breve{\beta}_1^\top &  \breve{\beta}_2\breve{\beta}_2^\top
        \end{bmatrix}^{-1}
        \begin{bmatrix}
            \mathcal{L}_x  \\
            \mathcal{L}_v
        \end{bmatrix}, 
    \end{align}
    where
    \begin{align*}
        &\mathcal{L}_x = \breve{\alpha}_1(z_p,t) + \breve{\beta}_1(z_p,t)(u_r+u_n+ d), \\
        &\mathcal{L}_v = \breve{\alpha}_2(z_p) + \breve{\beta}_2(z_p)(u_r+u_n+ d).
    \end{align*}
    From \eqref{eq:lambda_both_active}, the Gram matrix must be positive definite within the safe set (\cite{ames2017control}, Theorem 3), which ensures that $K_x \cap K_v$ is nonempty as $z_p$ approaches $\mathscr{C}_\cap^S(t)$.
    This condition is similarly required in our framework to guarantee that $u_s$ can simultaneously enforce both position and velocity constraints.
\end{remark}

\begin{remark}
    Our approach only requires the Gram matrix to be positive definite in a small neighborhood of $\mathscr{C}_\cap^S$, rather than in the whole set of $\mathscr{C}_\cap^I$ (cf. \cite{ames2017control}); {\hw thus enlarging the feasible solution set of the QP problem. }
    Furthermore, the conditions $L_{g_p}L_{f_p} b_x \neq0 ~\forall z_p \in \mathscr{C}_x$ and $L_{g_p} b_v \neq0 ~\forall z_p \in \mathscr{C}_v$ can be relaxed so that they only need to hold in a small neighborhood of the boundary of $\mathscr{C}_{x,1}$ and $\mathscr{C}_v$, respectively (\cite{max2023safe}, Assumption 1).
   {\hw These relaxations significantly benefit the applications of existing CBF based safe control approaches.}
\end{remark}


\begin{remark}
  {\wxy
  The robust safe control \eqref{eq:KKT_safe_control} guarantees the forward invariance of the safe set $\mathscr{C}_{x,1}(t) \cap \mathscr{C}_v$, but it requires accurate knowledge of the obstacle velocity $\partial x_o/\partial t$ and the disturbance $d$.
  In contrast, the robust safeguard policy \eqref{eq:safe_control}, as analyzed in Theorem \ref{thm:safe_analysis}, provides strict safety guarantees without requiring precise information about $d,~\frac{\partial x_o}{\partial t}$, or even their upper bounds $\bar{d}$ and $\eta$.
  Nevertheless, this relaxation comes at the expense of conservativeness, as \eqref{eq:safe_control} only enforces safety within the interior of the safe set, i.e., $z_p(t) \in \mathscr{C}_{\cap}^I(t)$ for all $t \geq 0$, rather than the entire set.
  }
\end{remark}

\section{Safe robust RL for safety-aware PE game} 
{\wxy 
Unconstrained nominal control policy $u_n$ relies on the solution of the HJI equation \eqref{eq:HJB_equation}, whose analytical solution is generally very difficult to obtain.
To address this challenge, this section proposes a neural network (NN) based safe robust RL framework that approximates the solution of \eqref{eq:HJB_equation} with low computational cost, as shown in Fig. \ref{fig:structure}}.


\begin{figure}[htbp]
    \centering
    \includegraphics[trim=0cm 0cm 0cm 0cm, clip, width=0.45\textwidth]{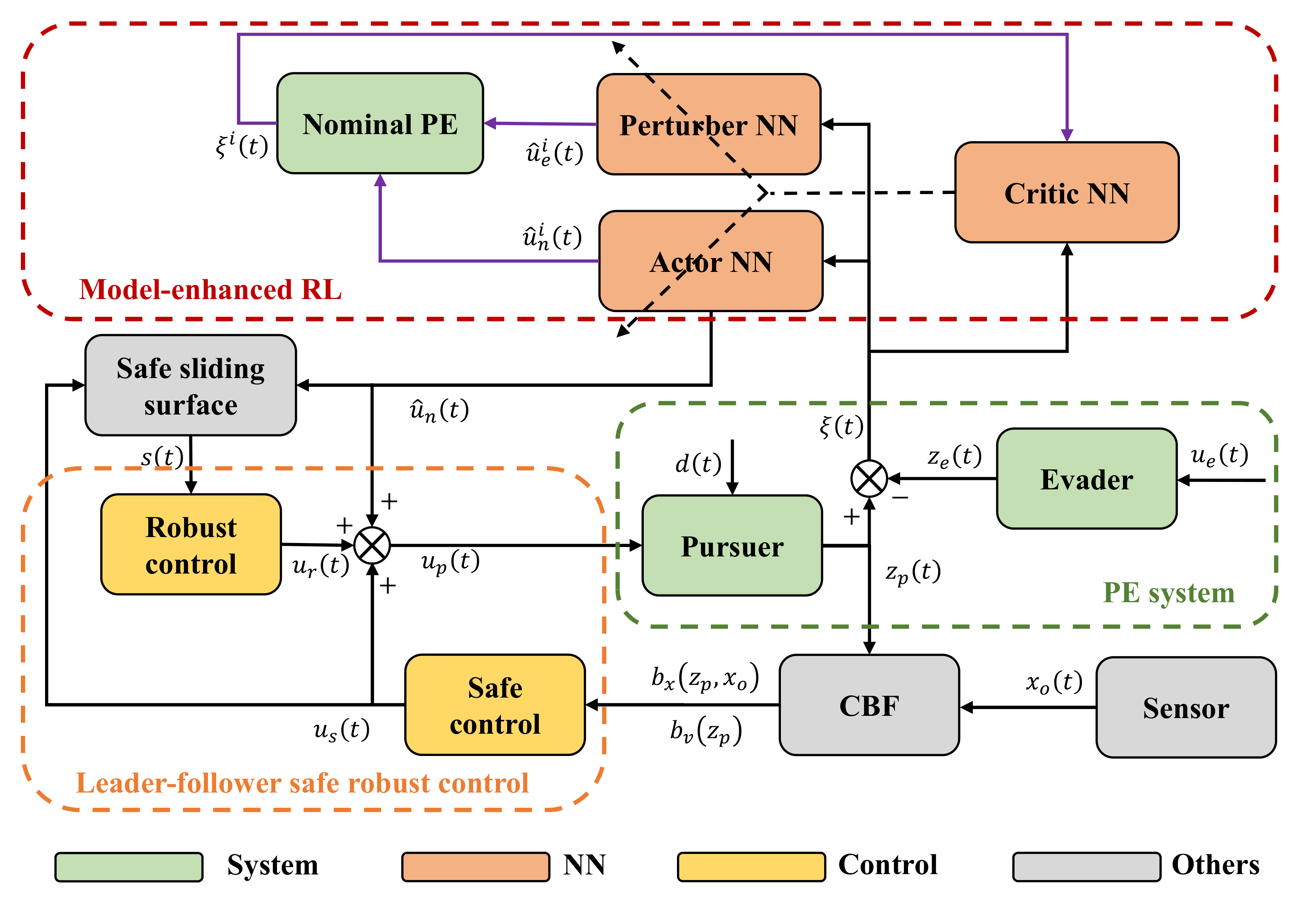}  
    \caption{Safe robust RL structure. Dash line represents back-propagating path, black solid line represents signal line, and purple solid line represents simulated data generation.}  \label{fig:structure}  
\end{figure}

\subsection{Actor-perturber-critic neural networks}
Under Assumption \ref{as:u_e}, one can always find a compact set $\Omega$ such that $\xi \in \Omega$ if $z_p \in \Omega_p$.
In the light of universal approximation theory, the optimal value function $V^\star(\xi)$, the optimal policies $u_n^\star(\xi)$ and $u_e^\star(\xi)$ can be expressed over such a compact set $\Omega$ as
 \begin{align} \label{eq:NN_expression}
     &V^\star(\xi) = W^\top \vartheta(\xi) + \epsilon_c(\xi) , \notag\\
    &u_n^\star(z_p,\xi) = -\frac{1}{2} R^{-1} g_p^\top(z_p) \big(\nabla \vartheta^\top(\xi) W +\nabla \epsilon_c(\xi)\big), \notag\\
    &u_e^{\star}(z_e,\xi) =-\frac{1}{2\gamma^2}T^{-1} g_e^\top(z_e) (\nabla \vartheta^\top(\xi) W+\nabla \epsilon_c(\xi)) ,
 \end{align}
where $W \in \mathbb{R}^{l_\vartheta}$ is an ideal NN weight satisfying $\lVert W \rVert \leq \bar{W}$, $\vartheta(\cdot): \mathbb{R}^{2n} \rightarrow \mathbb{R}^{l_\vartheta}$ is the continuously differentiable activation function with $l_\vartheta$ neurons, $\epsilon_c(\xi)$ is the function approximation error, $\nabla \vartheta= \partial \vartheta/\partial \xi$ and $\nabla \epsilon_c =\partial\epsilon_c/\partial \xi$.

Since the ideal weight $W$ is unknown, a critic NN, an actor NN and a perturber NN are used to approximate $V^\star(\xi)$, $u_n^\star(z_p,\xi)$ and $u_e^{\star}(z_e,\xi)$ as
\begin{align} \label{eq:policy_i}
    &\hat{V}(\xi) = \hat{W}_c^\top \vartheta(\xi), \notag\\
    &\hat{u}_n(z_p,\xi) = -\frac{1}{2} R^{-1} g_p^\top(z_p) \nabla \vartheta^\top (\xi)\hat{W}_a, \notag\\
    &\hat{u}_e(z_e,\xi) =-\frac{1}{2\gamma^2}T^{-1} g_e^\top(z_e) \nabla \vartheta^\top(\xi) \hat{W}_p,
\end{align}
where $\hat{W}_c$ is the critic NN weight, $\hat{W}_a$ is the actor NN weight and $\hat{W}_p$ is the perturber NN weight.

\subsection{Model-enhanced safe robust RL}
This section develops update laws for NNs to learn the optimal value function $V^\star$ and optimal control policies $u_n^\star ~\textnormal{and} ~u_e^\star$ while relaxing persistence condition using model information.
The Bellman error is defined as 
\begin{align} \label{eq:bellman_error}
    \delta(t) = \hat{W}_c^{\top}  \varphi(t) + l( \xi(\tau), \hat{u}_n(\tau),\hat{u}_e(\tau)),
\end{align}
where $\varphi=\nabla\vartheta^\top(\xi)(\mathcal{F}(z_p,z_e)+g_p(z_p)\hat{u}_n - g_e(z_e)\hat{u}_e)$.
%
%
To relax the persistence excitation condition required in classical RL, we generate simulated trajectories by extrapolating the current state to its neighborhood.
Denote $\{z_{p}^i(t),z_{e}^i(t),\xi_i(t)\}_{i=1}^N$ as a collection of $N$ state points sampled from $(z_p^i,z_e^i,\xi_i) \in \Omega_p \times \Omega_e \times \Omega$ at $t \geq 0$ under $\hat{u}^i_n := \hat{u}_n(z_p^i,\xi_i)$ and $\hat{u}^i_e := \hat{u}_e(z_e^i,\xi_i)$, where $\xi_i = z_{p}^i - z_{e}^i$.
The Bellman error at each sampled state $\delta_i:=\delta(\xi_i,z_p^i,z_e^i)$ is computed as
\begin{align} \label{eq:bellman_error_sample}
    \delta_i(t) = \hat{W}_c^{\top} \varphi_i(t) + l( \xi_i(t), \hat{u}^i_n(t),\hat{u}^i_e(t)), 
\end{align}
where $\varphi_i=\nabla\vartheta^\top(\xi_i)(\mathcal{F}(z_p^i,z_e^i)+ g_p(z_p^i)\hat{u}^i_n - g_e(z_e^i)\hat{u}^i_e)$.
%
The weight tuning law of the critic NN is 
\begin{align}
    \dot{\hat{W}}_c(t) &= -\Gamma(t) \Big(k_{c1} \frac{\varphi(t)}{\rho^2(t)} \delta(t) + \frac{k_{c2}}{N}\sum_{i=1}^N \frac{\varphi_i(t)}{\rho_i^2(t)} \delta_i(t)\Big) \label{eq:update_Wc_a} \\
    \dot{\Gamma}(t)& = \beta \Gamma(t) - \Gamma(t) \Big( k_{c1} \Lambda(t) + \frac{k_{c2}}{N}\sum_{i=1}^N \Lambda_i(t) \Big) \Gamma(t) \label{eq:update_Wc_b}
\end{align}
where $\rho(t) = 1 + \varphi(t)^{\top} \varphi(t)$, $\rho_i(t) = 1 + \varphi_i(t)^{\top} \varphi_i(t)$,
$$\Gamma(0) > 0,\ \Lambda(t) = \frac{\varphi(t) {\varphi(t)}^{\top}}{\rho^2(t)}\ \ \textnormal{and}\ \  \Lambda_i(t) = \frac{\varphi_i(t) {\varphi_i(t)}^{\top}}{\rho_i^2(t)},$$
are normalizing signals, $k_{c1},k_{c2}> 0$ are bounded learning gains, and $\beta>0$ is a forgetting factor. 

Let
\begin{align*}
    \mathcal{G}_{\vartheta}(\xi,z_p) &= \nabla\vartheta(\xi) g_p(z_p) R^{-1} g_p^\top(z_p)\nabla\vartheta^\top(\xi), \\
    \mathcal{H}_{\vartheta}(\xi,z_e) &=  \nabla\vartheta(\xi) g_e(z_e)  
T^{-1}g^\top_e(z_e)\nabla\vartheta^\top(\xi).
\end{align*}
Then, the weight tuning laws of the actor NN and the perturber NN  are  designed as 
\begin{align}
    \dot{\hat{W}}_a &= \operatorname{Proj}_{\mathcal{W}} \biggl\{-k_{a1}\left(\hat{W}_a - \hat{W}_c \right) - k_{a2} \hat{W}_a \nonumber \\
    & +  \frac{k_{c1}}{4 }\mathcal{G}_{\vartheta}\hat{W}_a \frac{\varphi^{\top}}{\rho^2}  \hat{W}_c + \frac{k_{c2}}{4 N}\sum_{i=1}^N \mathcal{G}_{\vartheta,i}\hat{W}_a \frac{\varphi_i^{\top}}{\rho_i^2}  \hat{W}_c \biggl\}, \label{eq:un_update_law_a}\\
    \dot{\hat{W}}_p &=  \operatorname{Proj}_{\mathcal{W}} \biggl\{-k_{p1}\left(\hat{W}_p - \hat{W}_c \right) - k_{p2} \hat{W}_p  \nonumber \\
    - &  \frac{k_{c1}}{4 \gamma^2}  \mathcal{H}_{\vartheta} \hat{W}_p \frac{\varphi_i^{\top}}{\rho_i^2}  \hat{W}_c - \frac{k_{c2}}{4 N \gamma^2}\sum_{i=1}^N  \mathcal{H}_{\vartheta,i} \hat{W}_p \frac{\varphi_i^{\top}}{\rho_i^2}  \hat{W}_c \biggl\}\label{eq:un_update_law_b},
\end{align}
where $k_{a1}, k_{a2},k_{p1} ~\textnormal{and}~ k_{p2} > 0$ are learning gains, $\mathcal{G}_{\vartheta,i}:= \mathcal{G}_\vartheta(\xi_i,z_p^i)$, $\mathcal{H}_{\vartheta,i}:= \mathcal{H}_\vartheta(\xi_i,z_e^i)$ and $\operatorname{Proj}_{\mathcal{W}} \{\cdot \}$ is a smooth operator for projecting $\hat{W}_a$ and  $\hat{W}_p$ within a compact set $\mathcal{W}$.

The detailed safe robust RL algorithm is presented in Algorithm \ref{Algorithm:RL} to online learn the Nash solution of PE game with safety guarantee and robustness.

\begin{algorithm}
\caption{Safe Robust RL for PE Differential Game}
\label{Algorithm:RL}
\begin{algorithmic}[1]
\Require $z_p(0)$, $\xi(0)$, $x_o(0)$, $t_{end}$ ($t_{end}$ is the running time)

\Ensure $\hat{W}_c(t_{end}), \hat{W}_a(t_{end}), \hat{W}_p(t_{end})$

\State \textbf{Initialization:} 
\Statex Game settings: $Q$, $R$, $T$, $\gamma$
\Statex Control gains: $K_1$, $K_2$, $K_3$, $K_4$, $\tilde{K}_x$, $\tilde{K}_v$, $\rho$
\Statex Learning weights: $k_{c1}$, $k_{c2}$, $k_{a1}$, $k_{a2}$, $k_{p1}$, $k_{p2}$, $N$
\Statex Adaptive weights: $\hat{W}_c(0)$, $\hat{W}_a(0) $, $\hat{W}_p(0)$, $\Gamma(0)$, $\mathcal{K}(0)$

\While{$t \leq t_{end}$}
    \State Update safe control policy $u_s(t)$ by \eqref{eq:safe_control}
    \State Update min-max policy $\hat{u}_n(t)$ by \eqref{eq:policy_i}
    \State Update robust control policy $u_r(t)$ by \eqref{eq:robust_control_tanh}
    \State Update pursuer's input $u_p(t) =u_s(t)+\hat{u}_n(t)+u_r(t)$
    \State  Collect real-time data \((\varphi, l(\xi, \hat{u}_n, \hat{u}_e))\)
    \State Generate simulated data \(\{(\varphi_i, l(\xi_i, \hat{u}_n^i, \hat{u}_e^i))\}_{i=1}^N\)
    \State \textbf{Policy evaluation:}
    \State \quad Compute bellman error $\delta(t)$ by \eqref{eq:bellman_error}
    \State \quad \textbf{for} $i=1,2,\cdots,N$ \textbf{do}
    \State \quad \quad Compute simulated bellman error $\delta_i(t)$ by \eqref{eq:bellman_error_sample}
    \State \quad \textbf{end for}
    \State \quad Update critic NN $\hat{W}_c(t)$ by \eqref{eq:update_Wc_a}
    \State \quad Update adaptive gain $\Gamma(t)$ by \eqref{eq:update_Wc_b}
    \State \textbf{Policy improvement:}
    \State \quad Update actor NN $\hat{W}_a(t)$ by \eqref{eq:un_update_law_a}
    \State \quad Update perturber NN $\hat{W}_p(t)$ by \eqref{eq:un_update_law_b}
\EndWhile
\end{algorithmic}
\end{algorithm}

\subsection{Safety and stability analysis}
In this section, we analyze the safety and stability of the closed-loop PE system under the evader's input $u_e$ and the pursuer's input
\begin{align} \label{eq:behvaior_policy}
    u_p =& \hat{u}_n + u_r + u_s \notag
    \\=&-\frac{1}{2} R^{-1}g_p^\top(z_p) \nabla \vartheta^\top \hat{W}_a- \mathcal{K}\tanh \left(\frac{g_p^\top(z_p) \chi^\top(z_p) s}{\rho}\right) \notag  \\
    &+\zeta_x(z_p,t) \breve{\beta}_1^\top(z_p,t) + \zeta_v(z_p) \breve{\beta}_2^\top(z_p).
\end{align}

{\wxy The effectiveness of the proposed control policy \eqref{eq:behvaior_policy} in enforcing safety is presented in the following proposition.}
\begin{proposition} \label{pro:safety}
    Consider the system \eqref{eq:pursuer_dynamics} and the safe set \eqref{eq:safetyset}.
    Let the control policy $u_p$ be designed as \eqref{eq:behvaior_policy}.
    Let $\zeta_x(z_p,t)$ and $\zeta_v(z_p)$ satisfy \eqref{bx_pro} and \eqref{bu_pro}, respectively.
    Suppose Assumptions \ref{as:disturbance}-\ref{as:u_e} hold, $z_p \in \Omega_p$ and $\xi \in \Omega$.
    %
    Then the safe robust control policy in \eqref{eq:behvaior_policy} can guarantee the safety of \eqref{eq:pursuer_dynamics}, i.e., $\phi_{x,1}(\xi, t), b_v(z_p) > 0 ~~\forall t \geq 0$.
\end{proposition}
\begin{proof}
%
Since $\hat{W}_a$ is bounded from \eqref{eq:un_update_law_a}, we denote the bound of $\hat{W}_a$ as $\bar{W}_a$.
For continuously differentiable $\vartheta$ and locally Lipschitz $g_p$, one has that $\vartheta$ and $g_p$ are bounded in the compact sets $\Omega$ and $\Omega_p$, respectively.
Then $\hat{u}_n$ is bounded as $\lVert \hat{u}_n \rVert \leq \frac{1}{2}\bar{\sigma}(R^{-1}) \overline{\lVert g_p(z_p)\rVert}_{\Omega_p}  \overline{\lVert \nabla \vartheta(\xi)\rVert}_{\Omega}  \bar{W}_a<\infty$.
From Assumptions \ref{as:disturbance}-\ref{as:u_e} and Lemma \ref{lm:K_upper_bound}, the signals $\breve{\beta}_1$, $\breve{\beta}_2$, $d$, and $u_r$ are bounded.
By redefining $P_x=L_{f_p} \pi + \frac{\partial \pi}{\partial x_o}^\top \frac{\partial x_o}{\partial t}+ \breve{\beta}_1(z_p,t)(u_r+\hat{u}_n+d)$ and $P_v = L_{f_p}b_v + \breve{\beta}_2(z_p)(u_r +\hat{u}_n+d)$, safety of \eqref{eq:pursuer_dynamics} can be shown by following the same development as that of Theorem \ref{thm:safe_analysis} based on the fact that $P_x$ and $P_v$ are bounded for all $z_p \in \Omega_p$ and bounded $x_o$.
\end{proof}

The following assumption is made to guarantee that the NN estimations $\hat{W}_c,~\hat{W}_a,~\textnormal{and}~\hat{W}_p$ converge to the ideal value.

\begin{assumption} \label{as:learning}  \cite{kamalapurkar2016model}
The sample number $N$ is large enough to satisfy the following condition 
        \begin{align*}
        \inf_{t\geq 0} \left\{\underline{\sigma} \left( \frac{1}{N} \sum_{i=1}^N \Lambda_i(t) \right)\right\} \geq  \lambda_c,
        \end{align*}
where $\lambda_c$ is a positive constant.
\end{assumption}

%
Define the estimation errors of NNs as $$\Tilde{W}_c = W - \hat{W}_c \ ,\Tilde{W}_a = W - \hat{W}_a \ \ \textnormal{and} \ \ \Tilde{W}_p = W - \hat{W}_p,$$  and let $z = \col\bm(\xi, \Tilde{W}_c, \Tilde{W}_a, \Tilde{W}_p \bm)$.
Now consider the following control Lyapunov function 
\begin{align}\label{eq:A3-Lyapunov_function}
    L(z,t) = V^\star(\xi) + V_c(\tilde{W}_c,t) + V_a(\tilde{W}_a) + V_p(\tilde{W}_p),
\end{align}
where $V_c = \frac{1}{2} \tilde{W}_c^{\top} \Gamma^{-1} \tilde{W}_c$, $V_a = \frac{1}{2} \lVert \tilde{W}_a\rVert^2$ and $V_p= \frac{1}{2} \lVert \tilde{W}_p\rVert^2$.
Under Assumption \ref{as:learning}, there exist positive constants $\bar{\Gamma}$ and $\underline{\Gamma}$ such that $\bar{\Gamma} I\leq  \Gamma(t) \leq \underline{\Gamma} I$, $\forall t \geq 0$ (\cite{rushikesh2016efficient}, Lemma 1).
Therefore, $L(z, t)\succ0$ and hence $\eta_1(\lVert z \rVert) \leq L(z, t) \leq \eta_2(\lVert z \rVert)$ holds for some class $\mathcal{K}$ functions $\eta_1$ and $\eta_2$ (\cite{khalil2002nonlinear}, Lemma 4.3).
The following theorem demonstrates the stability of the safe robust RL.

\begin{thm} \label{thm:convergence_RL}
    Consider the PE system \eqref{eq:pursuer_evader_dynamics}, the safe set $\mathscr{C}_p$ and $\mathscr{C}_v$. 
    Define a series of functions \eqref{eq:HO-DRCBF}, the corresponding sets $\mathscr{C}_{x,0}$ and $\mathscr{C}_{x,1}$ associated with the safe set $\mathscr{C}_x$.
    Design the control policy of pursuer as \eqref{eq:behvaior_policy} with $u_p$ defined in \eqref{eq:behvaior_policy}.
    Let the optimal value function and corresponding control policy be approximated by NNs over a compact set $\Omega$ with the critic NN updated by \eqref{eq:update_Wc_a} and \eqref{eq:update_Wc_b}, the actor and perturber NN updated by \eqref{eq:un_update_law_a} and \eqref{eq:un_update_law_b}, respectively.
    Suppose that Assumptions \ref{as:disturbance}-\ref{as:learning} hold, $z_p(0) \in \operatorname{Int}(\mathscr{C}_x) \cap \operatorname{Int}(\mathscr{C}_v)$, $\zeta_x(z_p,t)$ and $\zeta_v(z_p)$ satisfy \eqref{bx_pro} and \eqref{bu_pro}, respectively. 
    Then the trajectory $z(t)$ is uniformly ultimately bounded (UUB).
\end{thm}
\begin{proof}
    See Appendix \ref{Appendix_convergence_PI}. 
\end{proof}



\section{Simulation Example}
Consider the pursuer as a wheeled moving robot described by $\dot{x}_p = v_p, \dot{v}_p = f(x_p,v_p) + u_p + d$ and the evader robot as $\dot{x}_e = v_e, \dot{v}_e =f(x_e,v_e) + u_e$, where $f(\mu_1,\mu_2) = 0.1\cos(0.08 \mu_1)\sin(0.03\mu_2) -0.1\cos(0.03\mu_1)\tanh(0.08\mu_2)+0.05\mu_2$ represents the inherent nonlinearities; $u_p$ and $u_e$ are the control force of pursuer and evader robot, respectively;
 $x_p = \col(x_{p,1},x_{p,2}) \in \mathbb{R}^2$ and $v_p = \col(v_{p,1},v_{p,2})\in \mathbb{R}^2$, where $x_{p,1},v_{p,1}$ are the position and velocity along $X$ axis, and $x_{p,2},v_{p,2}$ are the position and velocity along $Y$ axis.
Then the dynamics of the PE error $\xi = \col(\xi_x,\xi_v)$ is 
\begin{align*}
         \dot{\xi}  =
    \begin{bmatrix}
         v_p-v_e\\
        f(x_p,v_p) - f(x_e,v_e)
    \end{bmatrix}
    +
    \begin{bmatrix}
        0 \\
        u_p+d
    \end{bmatrix}
    -
    \begin{bmatrix}
        0 \\
        u_e
    \end{bmatrix}.
\end{align*}

Recall $z_p=\col(x_p,v_p)$.
The safe set for position constraint of the pursuer is  $\mathscr{C}_x = \{z_p \in \mathbb{R}^4:~b_p(z_p,x_o)\geq 0\}$, where $b_p(z_p,x_o) = k_p(\lVert x_p -x_o  \rVert^2 - r_o^2) \geq 0$, $k_p$ is a positive number, and $r_o \in \mathbb{R}^2$ is the minimum safety distance between the pursuer and the obstacle.
The safe set for velocity constraint is  $\mathscr{C}_v = \{z_p  \in \mathbb{R}^4:~b_v(z_p)\geq 0\}$, where $b_v(z_p) = v_{p,1}-v_{\textnormal{min}}$ and $v_{\textnormal{min}} \in \mathbb{R}$ is the minimum velocity of $v_{p,1}$. Note that $v_{p,1}>0$ ($v_{p,1}<0$) means that the pursuer is heading right (left) along the X axis.
Design $$\tilde{b}_x = 2k_p(x_p-x_o)^\top v_p - \frac{k_p}{4k} \lVert x_p -x_o\rVert^2 -  kk_p\eta^2$$ and we obtain $\phi_{x,0} = b_x$ and $\phi_{x,1} = \tilde{b}_x + p_1\phi_{x,0}$.
Further, $\breve{\beta}_1(z_p,t)$ and $\breve{\beta}_2(z_p)$ can be calculated as $\breve{\beta}_1(z_p,t) = 2(x_p-x_o)^\top$ and $\breve{\beta}_2(z_p) = [1,0]$. 
According to \eqref{eq:safe_control}, the safe control policy can be designed as $u_s = -2\tilde{K}_x \nabla B_x(\phi_{x,1})(x_p-x_o)^\top + \tilde{K}_v\nabla B_v(b_v)$, where $B_x(\phi_{x,1}) = 1/\phi_{x,1}$ and $B_v(b_v) = 1/b_v$.

The minimum safety distance is $r_o = 5~ \textnormal{m}$, and the minimum velocity is $v_{\textnormal{min}} = -20 \textnormal{m/s}$.
The initial state of pursuer robot is $x_p(0) = \col(0,0), v_p(0) = \col(0,0)$ and the initial state of evader robot is $x_e(0) = \col(2,5), v_e(0) = \col(0,0)$.
The evader is tracking a prescribed trajectory, which satisfies Assumption \ref{as:u_e}.
%
%
The obstacle is starting from $x_o(0) = \col(10, 10)$ and moving towards its destination $x_{o,d}=\col(40,40)$ at a speed of $v_o= -k_{o}(x_o - x_{o,d})$.
The CBF parameters are set as $k_p = 0.1, ~p_1 = 50$, $k = 1$, $\tilde{K}_x =  1000$ and $\tilde{K}_v =  10$.
The parameters for adaptive sliding control are set as $K_1 = 1,~ K_2 = 10,~ K_3 = 10,~ K_4 = 10, ~\rho =0.1$.
The parameters for PE game is set as $Q(\xi) = \xi^\top \breve{Q} \xi$ with $\breve{Q} = R = T = I$ and $\gamma = 10$.
The activation function is selected as $\vartheta(\xi) = \col(\xi_x,\xi_v) \otimes\col(\xi_x,\xi_v)$, and the weights for NN is initialized as $\hat{W}_c(0) = \hat{W}_a(0)=\hat{W}_p(0) = \col(0.6076, 0.9674, 1.865,  0.7862, -0.9773,  0.5243, -0.5635,$ $0.6471, 3.0772,  -0.8426)$.
The disturbance $d= \col(d_1,d_2)$ where 
$d_1(t) = d_{w,1}(t) + 2\sin(5\pi t) + \cos(2\pi t)$, $d_2(t) = d_{w,2}(t) + \cos(10\pi t)$, $d_{w,1}$ and $d_{w,2}$ are random zero-mean disturbances of magnitude bounded by $\bar{d}_w = 5 \textnormal{m/s}^2$.

Let pursuer $1$ and pursuer $2$ be two pursuers implemented by safe robust RL (i.e., $u_p=u_s+u_r+\hat{u}_n$) and robust RL (i.e., $u_p=u_r+\hat{u}_n$), respectively.
The evader tracks a piece-wise continuous target position $x_{e,d}(t)$, where $x_{e,d}(t) = \col(38,10)$ for $t \in[0,6)\textnormal{s}$, $x_{e,d}(t) = \col(10,25)$ for $t \in[6,10)\textnormal{s}$, $x_{e,d}(t) = \col(42,36)$ for $t \in[10,15)\textnormal{s}$ and $x_{e,d}(t) = \col(10,48)$ for $t \geq 15\textnormal{s}$.
%

The position trajectories of PE game are given in Fig. \ref{fig:2D}.
%
%
{\hw The pursuer under safe robust RL makes a sharp turn when $t=7\textnormal{s}$ at the speed of $v_p=\col(-8.03,-2.91) ~\textnormal{m/s}$; and make another sharp turn when $t=10.5\textnormal{s}$ at the speed of $\col(2.7,5.94)~\textnormal{m/s}$.
This shows our control algorithm enables the pursuer to successfully avoid the obstacle at high speeds during the whole process and catches the evader at $t = 23.5\textnormal{s}$. However, when the pursuer is controlled by the robust RL algorithm, unfortunately, it hits the obstacle at $t=7.07\textnormal{s}$ and $t=10.76\textnormal{s}$.}
Figure \ref{fig:constraint} further demonstrates the effectiveness of our safe control policy.
Collision avoidance is achieved by our method while the pursuer under robust RL cannot avoid collision (see Fig. \ref{fig:constraint_pos}).
The pursuer under robust RL violates velocity constraint at $t = 6.5 \textnormal{s}$ and $t = 15.5 \textnormal{s}$, whereas the pursuer under safe robust RL keeps $v_{p,1} \geq v_{min}$ for all $t \geq 0$ (see Fig. \ref{fig:constraint_vel}).

\begin{figure*}[htbp]
    \centering
    \begin{subfigure}{0.23\textwidth}
        \includegraphics[width=\linewidth]{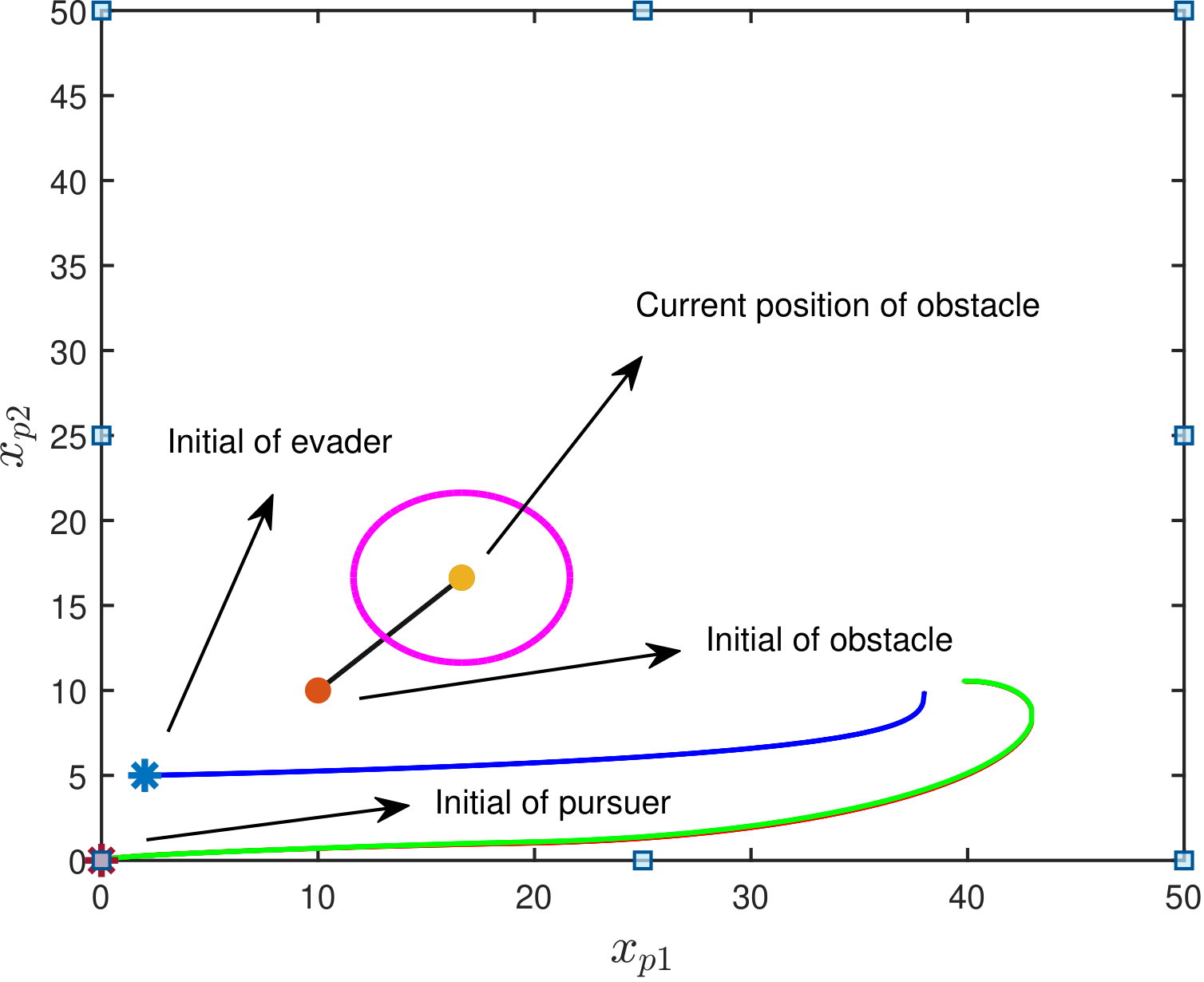}
        \caption{$t \in [0,5]\textnormal{s}$}
        \label{fig:a}
    \end{subfigure}
    \begin{subfigure}{0.23\textwidth}
        \includegraphics[width=\linewidth]{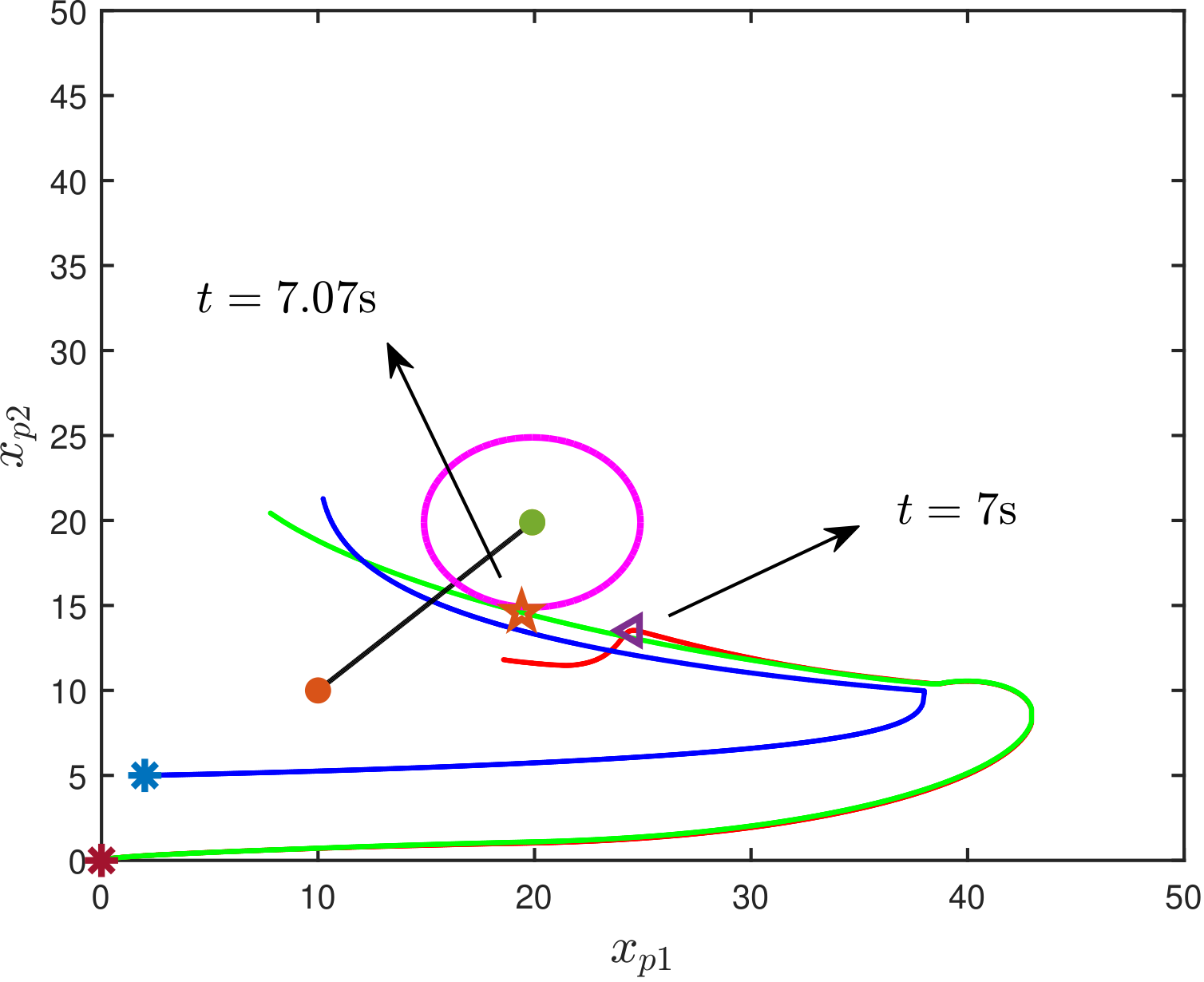}
        \caption{$t \in [0,8]\textnormal{s}$}
        \label{fig:b}
    \end{subfigure}
    \begin{subfigure}{0.23\textwidth}
        \includegraphics[width=\linewidth]{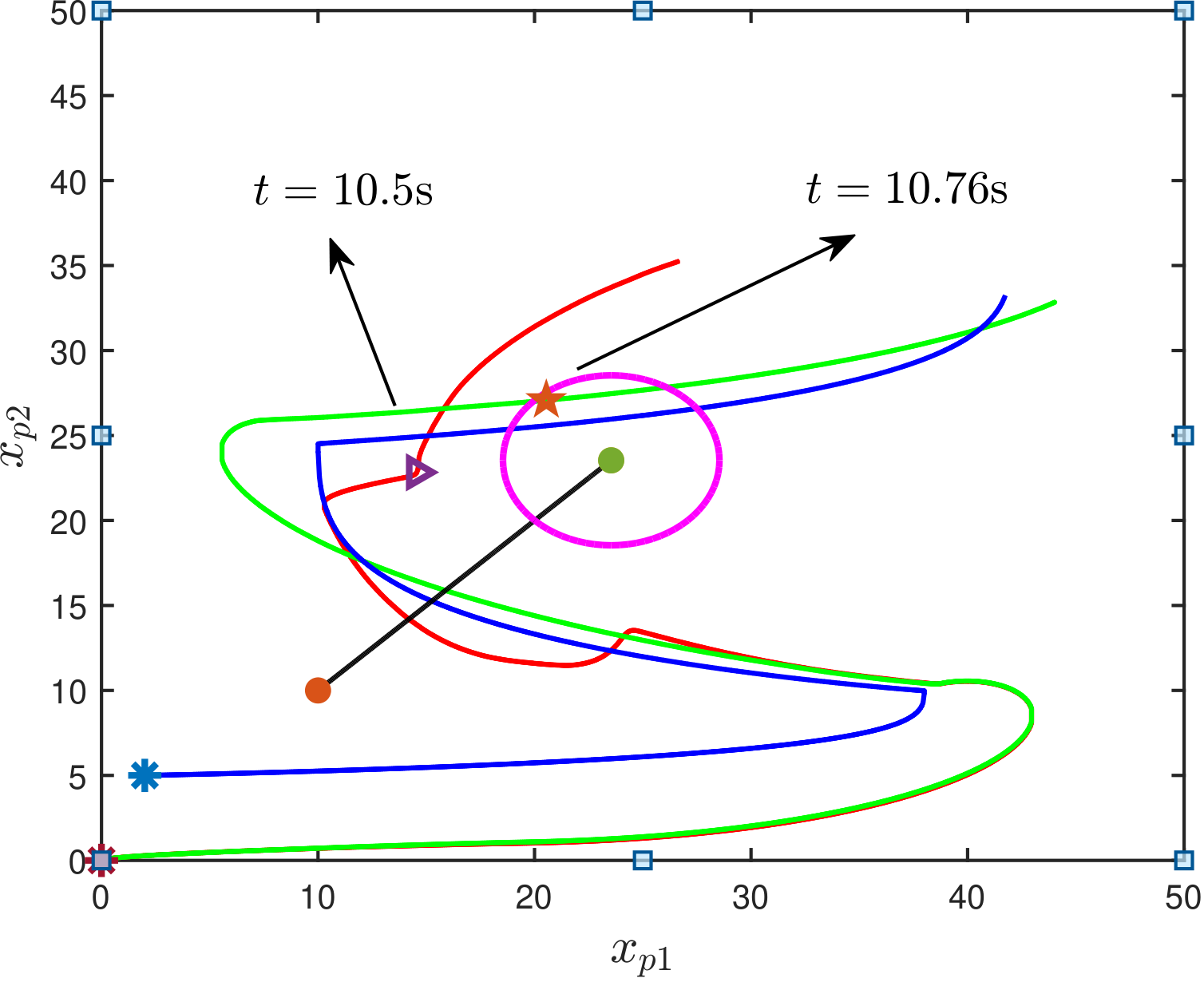}
        \caption{$t \in[0,12]\textnormal{s}$}
        \label{fig:c}
    \end{subfigure}
    \begin{subfigure}{0.23\textwidth}
        \includegraphics[width=\linewidth]{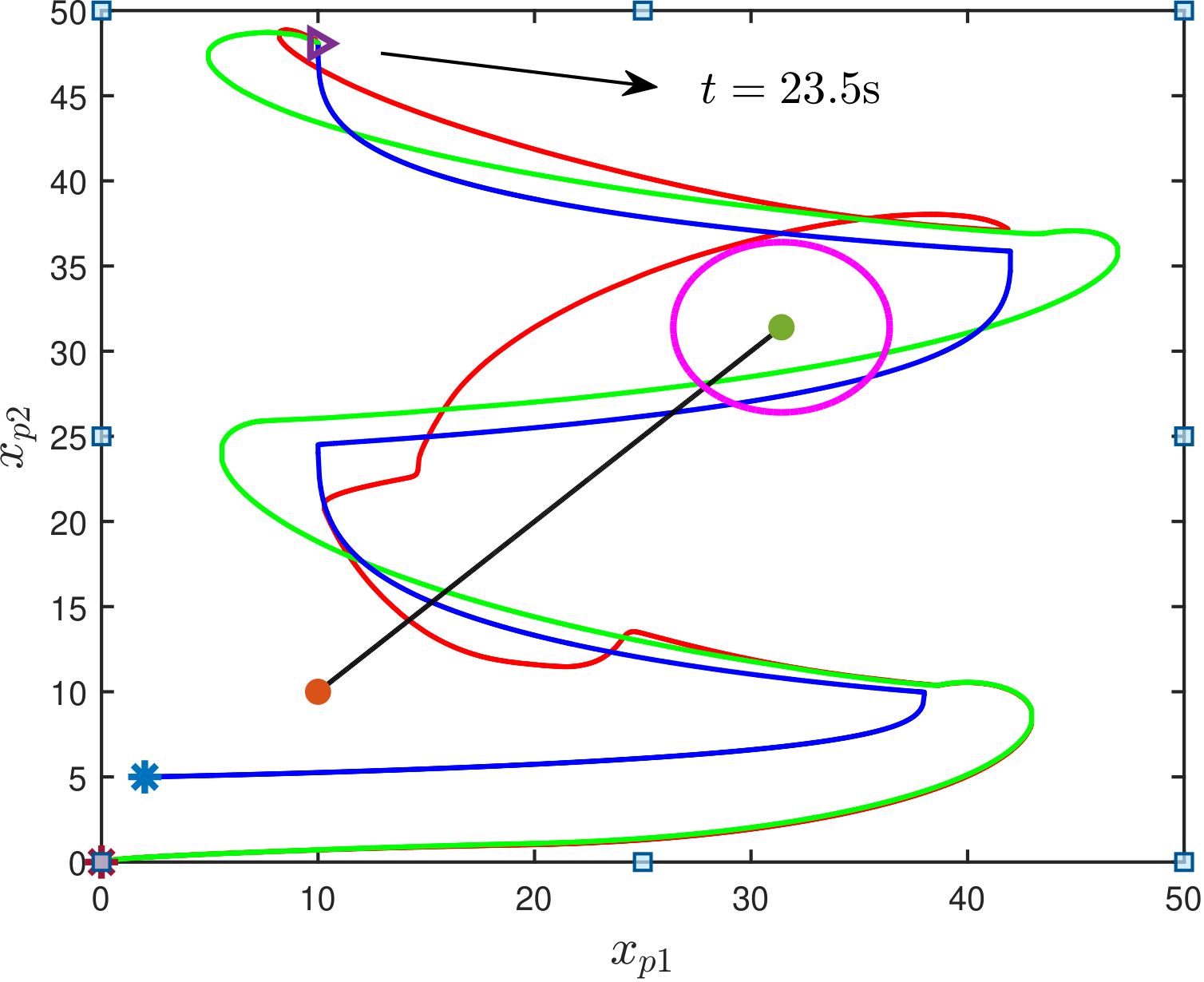}
        \caption{$t \in[0,25]\textnormal{s}$}
        \label{fig:d}
    \end{subfigure}
    \caption{{\hw Profiles of $b_p$ and $b_v$ when the pursuer is controlled by safe robust RL (red line) and robust RL (blue line) and  $b_p< 0$ means pursuer collides with the obstacle, and $b_v< 0$ means velocity constraint is violated.}}
    \label{fig:2D}
\end{figure*}

\begin{figure}[htbp]
    \centering
    \begin{subfigure}{0.48\textwidth}
        \includegraphics[width=\linewidth]{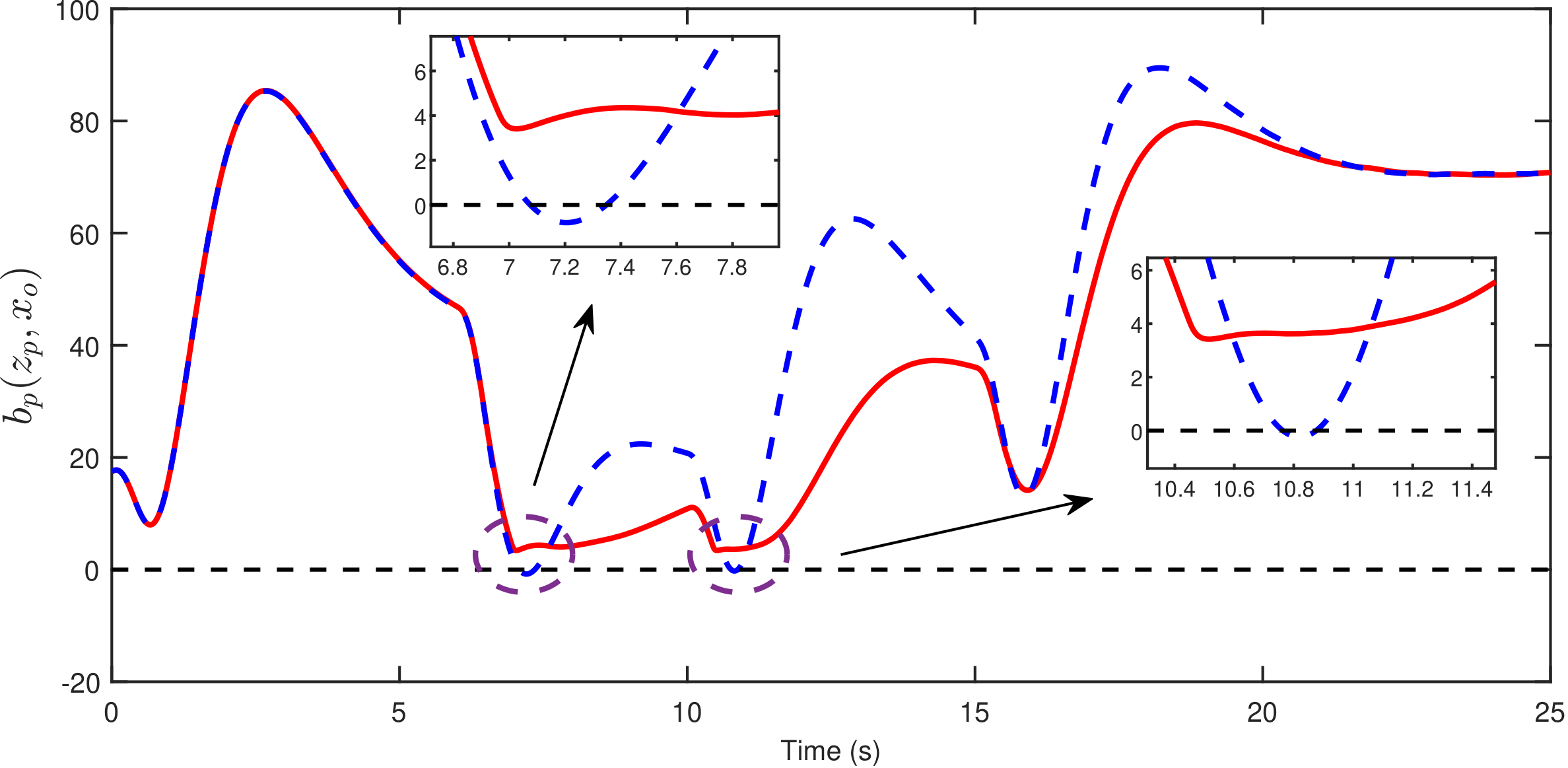}
        \caption{Position constraint}
        \label{fig:constraint_pos}
    \end{subfigure}
    \hfill
    \begin{subfigure}{0.48\textwidth}
        \includegraphics[width=\linewidth]{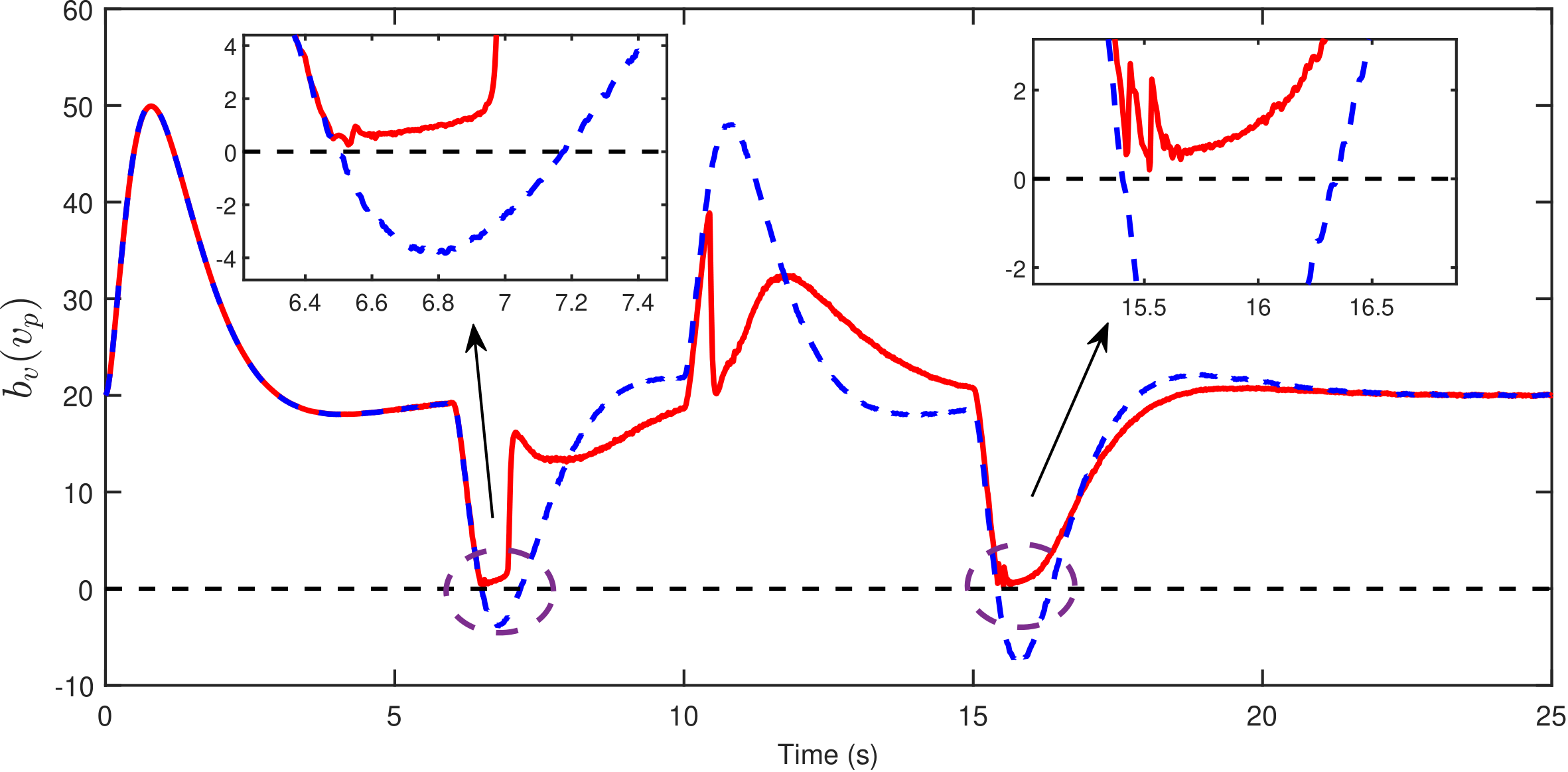}
        \caption{Velocity constraint}
        \label{fig:constraint_vel}
    \end{subfigure}
    \caption{{\hw Profiles of $b_p$ and $b_v$ when the pursuer is controlled by safe robust RL (red solid line) and robust RL (blue dash line) and  $b_p< 0$ means pursuer collides with the obstacle, and $b_v< 0$ means velocity constraint is violated.}}
    \label{fig:constraint}
\end{figure}




\section{Conclusion}
This paper proposes a robust safe RL framework for PE differential games {\wxy with unknown disturbances in an obstacle-rich environment}.
To enhance robustness, an adaptive sliding mode control (SMC) strategy {\wxy was} designed for disturbance rejection without knowing disturbance bound, while effectively mitigating the chattering issue typically caused by the sign function.
To guarantee safety, a robust safeguard policy {\wxy was} developed to enforce safety constraints in the presence of disturbances.
The robust safeguard policy and adaptive SMC strategy {\wxy were} then integrated within an RL framework to learn the Nash equilibrium of the PE game using both real-time data and simulated data.
%
{\wxy Future work will focus on extending our work to safety-aware multi-agent PE games.}

\appendices 
\section{Proof of Theorem \ref{thm:sliding_control}} \label{Appendix_sliding_control}

First, we study the case of $\lVert s(t)\rVert > \varepsilon$, when {\wxy the dynamics of $\mathcal{K}$} is given by \eqref{eq:sliding_adaptive_gain_update_s>0}.
    From Lemma \ref{lm:K_upper_bound}, there exists a positive number $\mathcal{K}^*$ such that $\mathcal{K}\leq \mathcal{K}^*$.
    Consider the following Lyapunov function candidate 
    \begin{align*}
        V_s(t) = \frac{1}{2}\lVert s(t)\rVert^2 + \frac{1}{2\breve{\gamma}}(\mathcal{K}(t)-\mathcal{K}^*)^2.
    \end{align*}
    Its time derivative along \eqref{eq:sliding_manifold}, \eqref{eq:sliding_adaptive_gain_update_s>0} is 
    \begin{align} \label{eq:Vs_dot}
        \dot{V}_s = &s^\top \chi(z_p)g_p(z_p)(u_r+d) + \frac{1}{\breve{\gamma}}(\mathcal{K}-\mathcal{K}^*)(K_1+K_2\lVert s \rVert) \nonumber\\
        = & -\mathcal{K} s^\top \chi(z_p) g_p(z_p) \tanh\left(\frac{g_p^\top(z_p) \chi^\top(z_p)s}{\rho}\right) \nonumber\\
        &+s^\top \chi(z_p)g_p(z_p)d+ \frac{1}{\breve{\gamma}}(\mathcal{K}-\mathcal{K}^*)(K_1+K_2\lVert s \rVert).
    \end{align}
    From Lemma \ref{lm:tanh_sign}, one has $$-s^\top \chi g_p \tanh(\frac{g_p^\top \chi^\top s}{\rho})\leq - \lVert g_p^\top \chi^\top s\rVert +\kappa m \rho.$$
    Let $\chi = g_p^\dagger$.
    According to the fact that $g_p^\dagger g_p = I$ and Assumption \ref{as:disturbance}, the inequality \eqref{eq:Vs_dot} becomes 
    \begin{align} \label{eq:V_s_dot}
        \dot{V}_s\leq & -\mathcal{K}\big(\lVert s\rVert - \kappa m \rho\big)+\bar{d} \lVert s\rVert \nonumber \\
        & +  \frac{1}{\breve{\gamma}}(\mathcal{K}-\mathcal{K}^*)(K_1+K_2\lVert s \rVert) \nonumber\\
        \leq & -(\mathcal{K}^*-\bar{d})\lVert s\rVert + \kappa m\rho  \mathcal{K} + (\mathcal{K}-\mathcal{K}^*)\frac{K_1}{\breve{\gamma}} \nonumber\\
        &+(\mathcal{K}-\mathcal{K}^*) \bigg(\frac{K_2}{\breve{\gamma}}-1 \bigg)\lVert s\rVert.
    \end{align}
    Note that $\mathcal{K} \leq \mathcal{K}^*$, and there always exists $\mathcal{K}^*$ and $\breve{\gamma}$ such that $\mathcal{K}^* > \bar{d}$ and $\breve{\gamma}< K_2$.
    Then \eqref{eq:V_s_dot} can be further written as
    \begin{align*}
        \dot{V}_s
        \leq & -(\mathcal{K}^*-\bar{d})\lVert s\rVert  - \frac{K_1}{\breve{\gamma}} |\mathcal{K}-\mathcal{K}^*|+ \kappa \mathcal{K}^*m\rho \\
        \leq & -\sqrt{2}\min\left\{\mathcal{K}^*-\bar{d}, \sqrt{\breve{\gamma}~}K_1\right\} V_s^{\frac{1}{2}} + \kappa\mathcal{K}^*m\rho.
    \end{align*}
    Denote $\beta_\gamma := \min\left\{\mathcal{K}^*-\bar{d}, \frac{K_1}{\sqrt{\breve{\gamma}~~}}\right\}$.
    Let $\varepsilon = \frac{\kappa\mathcal{K}^*m\rho}{\beta_\gamma}$.
    Then the sliding variable $s$ will reach a small domain $\Omega_s = \{s\in\mathbb{R}^m :~\lVert  s\rVert \leq \varepsilon\}$ in finite time.

    When $\lVert s \rVert \leq \varepsilon$, $\mathcal{K}(t)$ is determined by \eqref{eq:sliding_adaptive_gain_update_s=0}. 
    According to the proof of (\cite{plestan2010new}, Corollary 1), the states of \eqref{eq:pursuer_dynamics} will {\wxy remain} in $\Omega_s$.
    By selecting an appropriate $\rho$, the domain $\Omega_s$ can be arbitrary small.
    %
    {\wxy This further implies that $s = 0$ can be achieved in finite time.}

\section{Proof of Theorem \ref{thm:safe_analysis}} \label{Appendix_safe_analysis}
    %
    Define $B(t) = B_x(\phi_{x,1}(z_p, t)) + B_v(b_v(z_p))$.
    According to \eqref{bx_pro} and \eqref{bu_pro}, $B$ approaches $\infty$ if $z_p \rightarrow \mathscr{C}_{\cup}^S$.
    The proof is completed by demonstrating the non-existence of a time instant when $z_p \in \mathscr{C}_{\cup}^S$.
    The time derivative of $B$ along \eqref{eq:pursuer_dynamics} can be calculated as 
    \begin{align*}
        \dot{B}(t) = & \nabla B_x  P_x(z_p, t) -  \tilde{K}_x \lVert \breve{\beta}_1(z_p,t) \rVert^2 \lVert \nabla B_x\rVert^2    \\
        & + \nabla B_v P_v(z_p, t) -  \tilde{K}_v\lVert \breve{\beta}_2(z_p)\rVert^2 \lVert \nabla B_v\rVert^2 \\
        & - (\tilde{K}_x+\tilde{K}_v)\nabla B_x  \breve{\beta}_1(z_p,t) \breve{\beta}_2^\top (z_p)\nabla B_v,
    \end{align*}
    where $P_x(z_p,t) =  L_{f_p} \phi_{x,1} + \frac{\partial \phi_{x,1}}{\partial x_o}^\top \frac{\partial x_o}{\partial t} + \breve{\beta}_1(z_p,t)(u_r+u_n +d)$ and $P_v(z_p, t) = L_{f_p}b_v + \breve{\beta}_2(z_p)(u_r+u_n +d)$.

    Considering that $\phi_{x,1},~b_x$ are continuously differentiable, $f_p,~g_p$ are locally Lipschitz, and Assumption \ref{as:obstacle} holds, we have that $\lVert L_{f_p} \phi_{x,1} \rVert, \lVert L_{f_p} b_{v} \rVert, \lVert \breve{\beta}_1(z_p,t) \rVert,\lVert \breve{\beta}_2(z_p) \rVert~\textnormal{and}~ \lVert \frac{\phi_{x,1}}{\partial x_o} \rVert$ are bounded for any $z_p \in \Omega_p$ and any bounded $x_o$.
    Given the locally Lipschitz $u_n$ and Lemma \ref{lm:K_upper_bound}, $\lVert u_n\rVert$ and $\lVert u_r\rVert$ are bounded for any $z_p \in \Omega_p$ and $t \geq 0$.
    Recalling $\lVert d \rVert_\infty \leq \bar{d}$ and $\lVert \frac{\partial  x_o}{\partial t} \rVert \leq \eta$, one can conclude that $\lVert P_x(z_p,t)\rVert< \infty$ and $\lVert P_v(z_p, t)\rVert < \infty$ when $z_p$ approaches $\mathscr{C}_\cup^s$.
    
    Now, we will show that the trajectory of $z_p(t)$ cannot approach the boundary set $\mathscr{C}_\cup^S(t)$.
    In the following analysis, we divide the boundary set $\mathscr{C}_\cup^S(t)$ into three regions, $\mathscr{C}_\cap^S(t)$, $\mathscr{C}_{\cap,x}^S(t)$, $\mathscr{C}_{\cap, v}^S$, where $\mathscr{C}_{\cap,x}^S(t) = \mathscr{C}_\cup^S(t) \backslash \partial\mathscr{C}_v$ and $\mathscr{C}_{\cap, v}^S(t) = \mathscr{C}_\cup^S(t) \backslash \partial\mathscr{C}_{x,1}(t)$.
    \textbf{Case 1}: We consider the case that $z_p(t)$ approaches $\mathscr{C}^S_{\cap, x}(t)$, i.e., $\phi_{x,1}(z_p,t) \rightarrow 0$ and $b_v(z_p)>0$.
    In this case, $\dot{B}$ can be further written as 
    \begin{align*}
        \dot{B}(t) \leq &  \lVert \nabla B_x \rVert^2 \left(  \frac{ \lVert P_x(z_p,t)\rVert}{\lVert \nabla B_x \rVert }  +  \lVert \nabla B_v \rVert \frac{\lVert P_v(z_p,t) \rVert }{\lVert \nabla B_x \rVert^2}  \right. \\
         & \left.  - \tilde{K}_v\lVert \nabla B_v \rVert^2\frac{ \lVert \breve{\beta}_2(z_p)\rVert^2}{\lVert \nabla B_x \rVert^2} - \tilde{K}_x\lVert \breve{\beta}_1(z_p,t) \rVert^2  \right. \\
         & \left.  + (\tilde{K}_x+\tilde{K}_v) \lVert \nabla B_v \rVert \frac{ \lVert \breve{\beta}_1(z_p,t) \breve{\beta}^\top_2(z_p) \rVert }{\lVert \nabla B_x \rVert } \right).
    \end{align*}
    According to \eqref{bx_pro}, $\lVert \nabla B_x\rVert \rightarrow \infty$ as $\phi_{x,1}(z_p, t)$ approaches zero.
    Since $b_v > 0$, $b_v$ is continuously differentiable on $z_p$, $B_v$ is continuously differentiable on $b_v$, $b_v > 0$ implies $B_v < \infty$ and further $\lVert \nabla B_v \rVert < \infty$.
    Then one has $\lVert \nabla B_v \rVert < \lVert \nabla B_x \rVert$ in this case.
    Recalling that $\mathscr{C}_{\cap,x}^S \subset \mathscr{C}_\cup^S$, one has $\lVert P_x \rVert$ and $\lVert P_v \rVert$  are bounded when $z_p$ approaches $\mathscr{C}_{\cap,x}^S$, and further $\lim_{z_p \rightarrow \mathscr{C}^S_{\cap, x}(t)} \dot{B}(t) = -\infty$.

    \textbf{Case 2}: We consider the case that $z_p(t)$ approaches $\mathscr{C}^S_{\cap, v}(t)$, i.e., $\phi_{x,1}(z_p,t) > 0$ and $b_v(z_p) \rightarrow0$.
    In this case, $\dot{B}$ can be further written as 
    \begin{align*}
        \dot{B}(t) \leq &  \lVert \nabla B_v \rVert^2 \left( \lVert \nabla B_x \rVert  \frac{ \lVert P_x(z_p,t)\rVert}{\lVert \nabla B_v \rVert^2 }  +  \frac{\lVert P_v(z_p,t) \rVert }{\lVert \nabla B_v \rVert}  \right. \\
         & \left.  - \tilde{K}_v \lVert \breve{\beta}_2(z_p)\rVert^2 - \tilde{K}_x \lVert \nabla B_x \rVert^2\frac{\lVert \breve{\beta}_1(z_p,t) \rVert^2 }{\lVert \nabla B_v \rVert^2} \right. \\
         & \left.  + (\tilde{K}_x+\tilde{K}_v) \lVert \nabla B_x \rVert \frac{ \lVert \breve{\beta}_1(z_p,t) \breve{\beta}^\top_2(z_p) \rVert }{\lVert \nabla B_v \rVert } \right).
    \end{align*}
     Analogously, following the analysis in Case 1 and noting that $\mathscr{C}_{\cap,v}^S \subset \mathscr{C}_\cup^S$, it can be shown that $\lim_{z_p \rightarrow \mathscr{C}^S_{\cap, v}(t)} \dot{B}(t) = -\infty$.

    \textbf{Case 3}: We consider the case where $z_p(t)$ approaches $\mathscr{C}_{\cap}^S(t)$, i.e., $\phi_{x,1}(z_p, t), b_v(z_p) \rightarrow 0$.
    The time derivative of $B(t)$ is rewritten as
    \begin{align} \label{B-dot-x-u}
        \dot{B}(t) \leq & \lVert \nabla B_x \rVert^2 \lVert \nabla B_v \rVert^2 \left (  - \tilde{K}_x\frac{\lVert \breve{\beta}_1(z_p,t) \rVert^2}{\lVert \nabla B_v \rVert ^2}-   \tilde{K}_v\frac{\lVert \breve{\beta}_2(z_p) \rVert^2}{\lVert \nabla B_x \rVert^2} \right. \notag\\
        & \left.  + \frac{ \lVert P_x(z_p,t) \rVert}{\lVert \nabla B_x \rVert  \lVert \nabla B_v \rVert^2} + \frac{\lVert P_v(z_p,t) \rVert}{\lVert \nabla B_x \rVert^2  \lVert \nabla B_u \rVert} + \kappa \right),
    \end{align}
    where 
    \begin{align*}
        \kappa = -(\tilde{K}_x + \tilde{K}_v)\frac{ \breve{\beta}_1(z_p,t)  \breve{\beta}_2^\top(z_p) }{\nabla B_x \nabla B_v }.
    \end{align*}

    Denote the angle between $\breve{\beta}_1(z_p, t)$ and $\breve{\beta}_2(z_p)$ as $\theta$.
    Noting that $\breve{\beta}_1$ and $\breve{\beta}_2$ are row vectors, and $\breve{\beta}_1\breve{\beta}^\top_2 \in \mathbb{R}$.
    Using $\breve{\beta}_1(z_p, t) \breve{\beta}^\top_2(z_p) = \lVert\breve{\beta}_1(z_p, t) \rVert \lVert \breve{\beta}_2(z_p) \rVert \cos(\theta)$ (\cite{strang1988linear}, Chapter 3) and Young's Inequality, $\kappa$ can be further written as 
    \begin{align*}
    \kappa & = -(\tilde{K}_x + \tilde{K}_v )\cos(\theta) \frac{ \lVert \breve{\beta}_1(z_p,t)\rVert \lVert  \breve{\beta}_2(z_p)\rVert}{\nabla B_x  \nabla B_v } \\
    & \leq \frac{1}{2}(\tilde{K}_x + \tilde{K}_v ) \lVert \cos(\theta) \rVert \left(\frac{ \lVert \breve{\beta}_1(z_p,t) \rVert^2}{\lVert \nabla B_v \rVert  ^2} + \frac{ \lVert \breve{\beta}_2(z_p)\rVert^2 }{\lVert \nabla B_x \rVert^2} \right) .
    \end{align*}

    From Schur complement lemma, \eqref{schur-Matrix} holds if and only if $$\breve{\beta}_1\breve{\beta}_1^\top > 0,~~~\breve{\beta}_2\breve{\beta}_2^\top -  (\breve{\beta}_1\breve{\beta}_1^\top)^{-1}  \lVert \breve{\beta}_2\breve{\beta}_1^\top\rVert^2 > 0.$$
    %
    %
    Then this yields the following condition $$ \lVert \breve{\beta}_1 \rVert^2 \lVert\breve{\beta}_2 \rVert^2  - \lVert \breve{\beta}_2\breve{\beta}_1^\top\rVert^2 > 0.$$ Recalling $\breve{\beta}_2\breve{\beta}_1^\top = \lVert\breve{\beta}_1 \rVert \lVert \breve{\beta}_2\rVert\cos(\theta)$, we further have $$\lVert \breve{\beta}_1 \rVert^2 \lVert \breve{\beta}_2\rVert^2 \sin^2(\theta) >0.$$
    This implies that $\sin(\theta) \neq 0$, $\lVert \cos(\theta)\rVert < 1$ since $\cos(\theta) \in [-1,1]$,
    and there exists a positive constant $\epsilon_{\theta}$ such that $0 < \epsilon_{\theta} \leq 1 - \lVert \cos(\theta)\rVert$.
    Let $\breve{\epsilon}_{\theta,1} = \frac{1+\epsilon_{\theta}}{2} \tilde{K}_x - \frac{1-\epsilon_{\theta}}{2} \tilde{K}_v$ and $\breve{\epsilon}_{\theta,2} = \frac{1+\epsilon_{\theta}}{2} \tilde{K}_v - \frac{1-\epsilon_{\theta}}{2} \tilde{K}_x$.
    Then \eqref{B-dot-x-u} becomes
        \begin{align*}
        &\dot{B}(t) \leq  \lVert \nabla B_x \rVert^2 \lVert \nabla B_v \rVert^2 \left (   - \breve{\epsilon}_{\theta,1}  \frac{\lVert \breve{\beta}_1(z_p,t) \rVert^2}{\lVert \nabla B_v \rVert ^2}\right. \notag\\
         &\left. - \breve{\epsilon}_{\theta,2}  \frac{\lVert \breve{\beta}_2(z_p) \rVert^2}{\lVert \nabla B_x \rVert^2}  + \frac{ \lVert P_x(z_p,t) \rVert}{\lVert \nabla B_x \rVert  \lVert \nabla B_v \rVert^2} + \frac{\lVert P_v(z_p,t) \rVert}{\lVert \nabla B_x \rVert^2  \lVert \nabla B_v \rVert} \right).
    \end{align*}
    By selecting $\tilde{K}_x$ and $\tilde{K}_v$ such that $\breve{\epsilon}_{\theta,1}$ and $\breve{\epsilon}_{\theta,2} > 0$.
    Then one has $\lVert \nabla B_x\rVert$ and $\lVert \nabla B_v\rVert \rightarrow \infty$ as $\phi_{x,1}(z_p, t)$ and $ b_v(z_p)$ approach zero according to \eqref{bx_pro} and \eqref{bu_pro}.
    Recalling that $\mathscr{C}_{\cap}^S \subset \mathscr{C}_\cup^S$ and $\breve{\epsilon}_{\theta}  > 0$, one has $\lim_{z_p \rightarrow \mathscr{C}^S_{\cap}(t)} \dot{B}(t) = -\infty$.

    \textbf{Conclusion}: The above analysis yields 
    \begin{align} \label{B-dot-infty}
        \lim_{z_p \rightarrow \mathscr{C}^S_{\cup }(t)} \dot{B}(t) = - \infty.
    \end{align}
    Suppose there exists a finite time $T > 0$ such that $\lim_{t \rightarrow T} B(t) = \infty$, which further implies $\lim_{t \rightarrow T} \dot{B}(t) = \infty$.
    Since $B(t) \rightarrow \infty $ implies $ z_p \rightarrow \mathscr{C}^S_{\cup }(t)$, one has $\lim_{t \rightarrow T} z_p(t) \in \mathscr{C}^S_{\cup }(t)$ and further $\lim_{z_p \rightarrow \mathscr{C}^S_{\cup }(t)} \dot{B}(t) = \infty$, which contradicts with \eqref{B-dot-infty}.
    Hence, the non-existence of a finite $T$ precludes the existence of {\wxy such a} state trajectory $z_p$ {\wxy that} enters $\mathscr{C}_{\cup}^S(t)$.
    And one can conclude that the safe control policy \eqref{eq:safe_control} renders the safe set $\mathscr{C}_\cap^I(t)$ forward invariant.


\section{Proof of Theorem \ref{thm:convergence_RL}} \label{Appendix_convergence_PI}
Based on the HJB equation \eqref{eq:HJB_equation},
one can write the time derivative of $V^\star$ under \eqref{eq:behvaior_policy} and $u_e$ as 
    \begin{align}\label{eq:A3-Lyapunov}
        \dot{V}^\star = & \nabla V^{\star\top}(\mathcal{F}(z_p,z_e)+ g_p(z_p)(u_p+d) - g_e(z_e)u_e) \nonumber\\
        = & -Q(\xi)  - u_n^{\star\top} R u_n^\star + \gamma^2u_e^{\star\top} T u_e^\star - \nabla V^{\star\top}g_e(z_e) \tilde{u}_e \nonumber\\
        & +\nabla V^{\star\top}g_p(z_p)(u_s+u_r +d +\tilde{u}_n),
    \end{align}
where $\tilde{u}_n = \hat{u}_n +\frac{1}{2}R^{-1} g_p^\top \nabla V^\star$ and $\tilde{u}_e = u_e + \frac{1}{2\gamma^2}T^{-1}g_e^\top\nabla V^\star$.
Substituting \eqref{eq:NN_expression} in \eqref{eq:A3-Lyapunov} yields 
    \begin{align} \label{eq:A3-V*_dot}
        \dot{V}^\star =& -Q(\xi) - \frac{1}{4}W^\top (\mathcal{G}_{\vartheta} + \frac{\mathcal{H}_{\vartheta}}{\gamma^2})W  + \frac{1}{2}W^\top \mathcal{G}_{\vartheta}\tilde{W}_a\notag\\
        & + \frac{1}{4}\nabla\epsilon_c^\top (\mathcal{G}_{R} - \frac{\mathcal{H}_{T}}{\gamma^2}) \nabla\epsilon_c + \frac{1}{2}\nabla\epsilon_c^\top \mathcal{G}_R \nabla \vartheta^\top \tilde{W}_a\notag\\
        &+ W^\top \nabla \vartheta g_p (u_r+u_s + d) - W^\top \nabla \vartheta g_e u_e\notag\\
        \leq  -&Q(\xi)   + l_1 (\bar{u}_r+\bar{u}_s + \bar{d}) + l_2\bar{u}_e+ \bar{\epsilon}_1 + \bar{\epsilon}_2\lVert \tilde{W}_a \rVert,
    \end{align}
where $\bar{u}_r = \mathcal{K}^*$,
    \begin{align*}
        &l_1 = \bar{W}\overline{\lVert\nabla \vartheta(\xi) g_p(z_p)\rVert}_{\Omega\times \Omega_p},~~ l_2 = \bar{W}\overline{\lVert \nabla \vartheta(\xi) g_e (z_e)\rVert}_{\Omega\times\Omega_e},  \\
        &\bar{u}_s = \tilde{K} (\lVert \nabla B_x\rVert_{\infty}+\lVert \nabla B_v\rVert_{\infty}),~~\tilde{K} = \max\{\tilde{K}_x,\tilde{K}_v\},\\
        &\bar\epsilon_1 = \frac{1}{4}\overline{\lVert \nabla\epsilon_c^\top(\xi) \mathcal{G}_{R}(z_p) \nabla\epsilon_c(\xi)\rVert}_{\Omega\times \Omega_p}, ~~\mathcal{G}_R = g_p R^{-1}g_p^\top,\\
        &\bar{\epsilon}_2 = \frac{1}{2}\left(\bar{W} \overline{\lVert \mathcal{G}_{\vartheta}(\xi,z_p)\rVert}_{\Omega\times \Omega_p} + \overline{\lVert \nabla\epsilon_c^\top(\xi) \mathcal{G}_R(z_p) \nabla \vartheta^\top(\xi)\rVert}_{\Omega\times \Omega_p}\right).
    \end{align*}

    The HJB equation \eqref{eq:HJB_equation} can be also expressed by \eqref{eq:NN_expression} as 
    \begin{align*}
        0 =& W^\top \nabla\vartheta \mathcal{F}(z_p,z_e) +  Q(\xi)+ \epsilon_H(\xi) \\
        & - \frac{1}{4}W^\top \nabla\vartheta \mathcal{G}_R \nabla\vartheta^\top W +\frac{1}{4\gamma^2}W^\top \nabla\vartheta \mathcal{H}_T \nabla\vartheta^\top W, 
    \end{align*}
    where $\mathcal{H}_T = g_e  T^{-1}g_e^\top$ and 
    \begin{align*}
        \epsilon_H(\xi,z_p,z_e) = & \nabla\epsilon_c^\top\mathcal{F}(z_p,z_e)- \frac{1}{4}\nabla\epsilon_c^\top\mathcal{G}_R \nabla\epsilon_c +\frac{1}{4\gamma^2}\nabla\epsilon_c^\top \mathcal{H}_T \nabla\epsilon_c \\
        &- \frac{1}{2}\nabla\epsilon_c^\top\mathcal{G}_R \nabla\vartheta^\top W +\frac{1}{2\gamma^2}\nabla\epsilon_c^\top\mathcal{H}_T \nabla\vartheta^\top W.
    \end{align*}

    The Bellman error \eqref{eq:bellman_error} at $(\xi,z_p,z_e)$ can be rewritten as
    \begin{align*}
        \delta(t) =-\tilde{W}_c^\top \varphi(t) + \Pi(t),
    \end{align*}
    where 
    \begin{align*}
        \Pi(t)=\frac{1}{4}\tilde{W}_a^\top \mathcal{G}_{\vartheta}\tilde{W}_a - \frac{1}{4\gamma^2}\tilde{W}_p^\top \mathcal{H}_{\vartheta}\tilde{W}_p + \epsilon_{H}.
    \end{align*}
    The Bellman error \eqref{eq:bellman_error_sample} at $(\xi_i,z_p^i,z_e^i)$ can be also rewritten as
    \begin{align*}
        \delta_i(t) = -\tilde{W}_c \varphi_i(t) + \Pi_i(t),
    \end{align*}
    where 
    $$
    \Pi_i(t)=\frac{1}{4}\tilde{W}_a^\top \mathcal{G}_{\vartheta,i}\tilde{W}_a - \frac{1}{4\gamma^2}\tilde{W}_p^\top \mathcal{H}_{\vartheta,i}\tilde{W}_p + \epsilon_{H,i}.$$

    The time derivative of $V_c$ is
    \begin{align*}
        \dot{V}_c = & \tilde{W}_c^\top \Gamma^{-1} \dot{\tilde{W}}_c - \frac{1}{2}\tilde{W}_c^\top \Gamma^{-1} \dot{\Gamma} \Gamma^{-1}\tilde{W}_c \\
        = &-\frac{1}{2}\tilde{W}_c^\top\left( \beta \Gamma^{-1}+k_{c1} \Lambda + \frac{k_{c2}}{N} \sum_{i=1}^N \Lambda_i\right) \tilde{W}_c  \\
        & + \tilde{W}_c^\top \left(k_{c1} \frac{\varphi}{\rho^2} \Pi + \frac{k_{c2}}{N} \sum_{i=1}^N \frac{\varphi_i}{\rho_i^2}\Pi_i\right).
    \end{align*}
    Since $f_p$ and $f_e$ are locally Lipschitz, $\lVert \epsilon_H(\xi,z_p,z_e)\rVert$ is bounded for any $(\xi,z_p,z_e) \in \Omega \times \Omega_p \times \Omega_e$.
    Define $\bar{\epsilon}_H = \overline{\lVert\epsilon_H(\xi,z_p,z_e)\rVert}_{\Omega \times \Omega_p \times \Omega_e}$, $\bar{\mathcal{G}}_{\vartheta} = \overline{\lVert\mathcal{G}_\vartheta(\xi,z_p) \rVert}_{\Omega \times \Omega_p}$ and $\bar{\mathcal{H}}_{\vartheta} = \overline{\lVert\mathcal{H}_\vartheta (\xi,z_e)\rVert}_{\Omega \times \Omega_e}$.
    %
    %
    %
    %
    %
    Considering that $\frac{\varphi}{\rho^2} \leq \bar{\varphi}$ for a bounded positive value $\bar{\varphi}$, we have
    \begin{align}\label{eq:A3-Vc_dot}
        \dot{V}_c \leq&  -\bar{k}_c \bar{\lambda} \lVert \tilde{W}_c\rVert^2  +  \frac{1}{4}\tilde{W}_c^\top\bar{\Pi}_a 
        \notag\\
        &- \frac{1}{4\gamma^2}\tilde{W}_c^\top\bar{\Pi}_p+ \frac{1}{4} \bar{k}_c \bar{\varphi} \bar{\epsilon}_H \lVert \tilde{W}_c\rVert,
    \end{align}
    where $\bar{\lambda} = \frac{\beta}{2\bar{k}_c \lVert \Gamma\rVert_{\infty}} + \frac{\lambda_c k_{c2}}{2\bar{k}_c}$, $\bar{k}_c=k_{c1}+k_{c2}$  and
    \begin{align*}
        \bar{\Pi}_a = k_{c1} \frac{\varphi}{\rho^2}\tilde{W}_a^\top \mathcal{G}_{\vartheta}\tilde{W}_a +\frac{k_{c2}}{N} \sum_{i=1}^N \frac{\varphi_i}{\rho_i^2} \tilde{W}_a^\top \mathcal{G}_{\vartheta,i}\tilde{W}_a\\
        \bar{\Pi}_p = k_{c1} \frac{\varphi}{\rho^2}\tilde{W}_p^\top \mathcal{H}_{\vartheta}\tilde{W}_p +\frac{k_{c2}}{N} \sum_{i=1}^N \frac{\varphi_i}{\rho_i^2} \tilde{W}_p^\top \mathcal{H}_{\vartheta,i}\tilde{W}_p .
    \end{align*}

    Define $\bar{k}_a = k_{a1}+k_{a2}$. 
    The time derivative of $V_a$ is 
    \begin{align} \label{eq:A3-Va_dot}
        \dot{V}_a \leq & k_{a1} \tilde{W}_a^\top (\hat{W}_a - \hat{W}_c) + k_{a2} \tilde{W}_a^\top \hat{W}_a \notag\\
        &- \frac{k_{c1}}{4}  \tilde{W}_a^\top \mathcal{G}_{\vartheta} \hat{W}_a \frac{\varphi^\top}{ \rho^2} \hat{W}_c- \frac{k_{c2}}{4N}\sum_{i=1}^N  \tilde{W}_a^\top \mathcal{G}_{\vartheta,i} \hat{W}_a \frac{\varphi_i^\top}{ \rho_i^2} \hat{W}_c \notag\\
        = & - \bar{k}_a \tilde{W}_a^\top \tilde{W}_a +  k_{a1} \tilde{W}_a^\top \tilde{W}_c+ k_{a2} \tilde{W}_a^\top W\notag\\
        & - \frac{k_{c1}}{4}  \tilde{W}_a^\top \mathcal{G}_{\vartheta}\left(W \frac{\varphi^\top}{\rho^2}W +\tilde{W}_a \frac{\varphi^\top}{\rho^2}\tilde{W}_c\right) \notag\\
        & +\frac{k_{c1}}{4}  \tilde{W}_a^\top \mathcal{G}_{\vartheta} \left(\tilde{W}_a \frac{\varphi^\top}{\rho^2}W + W \frac{\varphi^\top}{\rho^2} \tilde{W}_c\right)\notag \\
        & - \frac{k_{c2}}{4N} \sum_{i=1}^{N} \tilde{W}_a^\top \mathcal{G}_{\vartheta,i}\left(W \frac{\varphi_i^\top}{\rho_i^2}W +\tilde{W}_a \frac{\varphi_i^\top}{\rho_i^2}\tilde{W}_c\right) \notag\\
        & +\frac{k_{c2}}{4N} \sum_{i=1}^{N} \tilde{W}_a^\top \mathcal{G}_{\vartheta,i} \left(\tilde{W}_a \frac{\varphi_i^\top}{\rho_i^2}W + W \frac{\varphi_i^\top}{\rho_i^2} \tilde{W}_c\right).
    \end{align}

    Similarly, the time derivative of $V_p$ is 
    \begin{align} \label{eq:A3-Vp_dot}
        \dot{V}_p \leq  & - \bar{k}_p \tilde{W}_p^\top \tilde{W}_p +  k_{p1} \tilde{W}_p^\top \tilde{W}_c+ k_{p2} \tilde{W}_p^\top W\notag\\
        & + \frac{k_{c1}}{4\gamma^2}  \tilde{W}_p^\top \mathcal{H}_{\vartheta}\left(W \frac{\varphi^\top}{\rho^2}W +\tilde{W}_p \frac{\varphi^\top}{\rho^2}\tilde{W}_c\right) \notag\\
        & - \frac{k_{c2}}{4\gamma^2} \tilde{W}_p^\top \mathcal{H}_{\vartheta} \left(\tilde{W}_p \frac{\varphi^\top}{\rho^2}W + W \frac{\varphi^\top}{\rho^2} \tilde{W}_c\right)\notag\\
        & + \frac{k_{c2}}{4N\gamma^2} \sum_{i=1}^{N} \tilde{W}_p^\top \mathcal{H}_{\vartheta,i}\left(W \frac{\varphi_i^\top}{\rho_i^2}W +\tilde{W}_p \frac{\varphi_i^\top}{\rho_i^2}\tilde{W}_c\right) \notag\\
         -&\frac{k_{c2}}{4N\gamma^2} \sum_{i=1}^{N} \tilde{W}_p^\top \mathcal{H}_{\vartheta,i} \left(\tilde{W}_p \frac{\varphi_i^\top}{\rho_i^2}W + W \frac{\varphi_i^\top}{\rho_i^2} \tilde{W}_c\right),
    \end{align}
    where $\bar{k}_p = k_{p1}+k_{p2}$.

    For brevity, we introduce the following symbols:
    \begin{align*}
        &\mu_a = \frac{\bar{k}_c}{4}  \bar{\varphi}\bar{\mathcal{G}}_{\vartheta}  \bar{W},~~~\mu_p = \frac{\bar{k}_c}{4\gamma^2}  \bar{\varphi}\bar{\mathcal{H}}_{\vartheta}  \bar{W},~~~~~l_c = \frac{1}{4} \bar{k}_c \bar{\varphi} \bar{\epsilon}_H,\\
        &l_a = k_{a2} \bar{W}+\frac{\bar{k}_c}{4} \bar{\mathcal{G}}_{\vartheta} \bar{\varphi} \bar{W}^2,~~~~~l_p = k_{p2} \bar{W}+\frac{\bar{k}_c}{4\gamma^2}  \bar{\mathcal{H}}_{\vartheta} \bar{\varphi} \bar{W}^2,\\
        &l_{ac} = k_{a1}+\frac{\bar{k}_c}{4}  \bar{\varphi}\bar{\mathcal{G}}_{\vartheta} \bar{W} ,~~~~~l_{pc} =k_{p1} + \frac{\bar{k}_c}{4\gamma^2}  \bar{\varphi}\bar{\mathcal{H}}_{\vartheta} \bar{W}.
    \end{align*}
    Then
    \begin{align} \label{eq:A3-W_dot}
        \dot{V}_c+\dot{V}_a+\dot{V}_p \leq &  -\bar{k}_c \bar{\lambda} \lVert \tilde{W}_c\rVert^2 + l_c\lVert \tilde{W}_c\rVert - (\bar{k}_a- \mu_a) \lVert \tilde{W}_a\rVert^2  \notag \\
         & + l_p \lVert \tilde{W}_p\rVert + l_a \lVert \tilde{W}_a \rVert + l_{pc}\lVert \tilde{W}_c\rVert\lVert \tilde{W}_p\rVert \notag \\
         & - (\bar{k}_p - \mu_p) \lVert \tilde{W}_p \rVert^2+ l_{ac}\lVert \tilde{W}_c\rVert\lVert \tilde{W}_a\rVert, 
    \end{align}
    %
    Considering \eqref{eq:A3-V*_dot}, \eqref{eq:A3-W_dot} and \eqref{eq:A3-Lyapunov_function}, we have
    \begin{align*}
        \dot{L} \leq &  -Q(\xi) - \bar{k}_c \bar{\lambda} \lVert \tilde{W}_c\rVert^2- \left(\bar{k}_a- \mu_a\right) \lVert \tilde{W}_a\rVert^2  \\
        &- \left(\bar{k}_p - \mu_p\right) \lVert \tilde{W}_p \rVert^2+ l_c \lVert \tilde{W}_c\rVert + (l_a+\bar{\epsilon}_2) \lVert \tilde{W}_a \rVert  \\
        &   + l_p \lVert \tilde{W}_p\rVert+ l_{ac} \lVert \tilde{W}_c\rVert\lVert \tilde{W}_a\rVert  + l_{pc}\lVert \tilde{W}_c\rVert\lVert \tilde{W}_p\rVert \\
        &   + l_1 (\bar{u}_r+\bar{u}_s + \bar{d}) + l_2\bar{u}_e+ \bar{\epsilon}_1 \\
        \leq &  -Q(\xi) -\frac{\bar{k}_c \bar{\lambda}}{2} \lVert \tilde{W}_c\rVert^2- \frac{\bar{k}_a- \mu_a}{2} \lVert \tilde{W}_a\rVert^2  \\
        &- \frac{\bar{k}_p - \mu_p}{2}\lVert \tilde{W}_p \rVert^2+ l_{ac} \lVert \tilde{W}_c\rVert\lVert \tilde{W}_a\rVert  + l_{pc}\lVert \tilde{W}_c\rVert\lVert \tilde{W}_p\rVert + \varrho  \\
        =&  -Q(\xi) -\frac{\bar{k}_c \bar{\lambda}}{4} \lVert \tilde{W}_c\rVert^2- \frac{\bar{k}_a- \mu_a}{4} \lVert \tilde{W}_a\rVert^2  \\
        &- \frac{\bar{k}_p - \mu_p}{4}\lVert \tilde{W}_p \rVert^2- y^\top M y + \varrho,
    \end{align*}
    where $y = \col(\lVert \tilde{W}_c\rVert,\lVert \tilde{W}_a\rVert,\lVert \tilde{W}_p\rVert)$, $\varrho = l_1 (\bar{u}_r+\bar{u}_s + \bar{d}) + l_2\bar{u}_e+ \bar{\epsilon}_1  + \frac{l_c^2}{2\bar{k}_c\bar{\lambda}} + \frac{(l_a +\bar{\epsilon})^2}{2(\bar{k}_a - \mu_a)}+ \frac{l_p^2}{2(\bar{k}_p - \mu_p)}$ and 
    \begin{align*}
        M = \begin{bmatrix}
            \frac{\bar{k}_c\bar{\lambda}}{4} & -\frac{l_{ac}}{2} & -\frac{l_{pc}}{2}\\
            -\frac{l_{ac}}{2} &  \frac{\bar{k}_a- \mu_a}{4} & 0\\
            -\frac{l_{pc}}{2} & 0 & \frac{\bar{k}_p- \mu_p}{4}
        \end{bmatrix}.
    \end{align*}
    Design $k_{c1}, k_{c2}, k_{a1}, k_{a2}, k_{p1}, k_{p2}$ and $\tilde{K}$ such that 
    \begin{align*}
        \bar{k}_a \geq \mu_a + 2l_{ac} ,~~~~ \bar{k}_p \geq \mu_p + 2l_{pc} ,~~~~ \frac{\bar{k}_c\bar{\lambda}}{2} \geq l_{ac} + l_{pc}.
    \end{align*}
    Then $M$ is positive semi-definite according to Geršgorin circle criterion.
    If $Q(\xi) \geq k_\xi \lVert \xi \rVert^2$ for some positive $k_\xi$, then 
    \begin{align} \label{eq:A3-L_dot}
        \dot{L} \leq & -k_\xi \lVert \xi \rVert^2-\frac{\bar{k}_c\bar{\lambda}}{4} \lVert \tilde{W}_c\rVert^2- \frac{\bar{k}_a- \mu_a}{4} \lVert \tilde{W}_a\rVert^2  \notag\\
        &- \frac{\bar{k}_p - \mu_p}{4}\lVert \tilde{W}_p \rVert^2 + \varrho,  \notag\\
    \leq & - \underbrace{\operatorname{min}\Big \{k_{\xi},\frac{\bar{k}_c \bar{\lambda}}{4}, \frac{\bar{k}_a- \mu_a}{4},\frac{\bar{k}_p - \mu_p}{4}\Big\}}_{\kappa>0}\lVert z \rVert^2 + \varrho,
    \end{align}
    which implies that 
    \begin{align*}
        \dot{L} \leq - \frac{\kappa}{2} \lVert z \rVert^2, & &  \forall \lVert z \rVert  \geq \sqrt{\frac{2\varrho}{\kappa}}.
    \end{align*}
    Recalling $L \succ 0$, the tracking error $\xi$, the NN weight estimations $\tilde{W}_c, \tilde{W}_a$ and $\tilde{W}_p$ are UUB (\cite{khalil2002nonlinear}, Theorem 4.18).

\bibliographystyle{IEEEtran}        
\bibliography{xywang}          
\end{document}